\documentclass[smallextended]{svjour3}

\usepackage{multido}
\usepackage{graphicx,float}
\usepackage{amssymb}
\usepackage{hyperref}
\usepackage{enumerate}

\usepackage{latexsym}
\usepackage{amsfonts}
\usepackage{amscd,amssymb,amsmath,array,hhline}

\numberwithin{equation}{section}
\numberwithin{definition}{section}
\numberwithin{lemma}{section}
\numberwithin{proposition}{section}
\numberwithin{theorem}{section}
\numberwithin{corollary}{section}

\numberwithin{figure}{section}

\newcommand{\E}{\mathsf{E}}
\renewcommand{\P}{\mathsf{P}}
\newcommand{\Q}{\mathsf{Q}}
\newcommand{\QQ}{\mathbb{Q}}
\newcommand{\Acc}{\cC}
\newcommand{\Att}{\mathsf{A}}

\newcommand\cF{{\mathcal F}}
\newcommand\cL{{\mathcal L}}
\newcommand\cB{{\mathcal B}}
\newcommand\cC{{\mathcal C}}
\newcommand\cH{{\mathcal H}}
\newcommand\cM{{\mathcal M}}

\newcommand\cZ{{\mathcal Z}}

\newcommand{\XX}{\mathfrak{X}}
\renewcommand{\XX}{\R^d}
\newcommand{\R}{\mathbb{R}}

\newcommand{\Lnot}[1][0]{\cL^#1}
\newcommand{\barLnot}[1][0]{{\overline\cL}^#1}
\newcommand{\Lts}[1][s]{\cL^p_{t,#1}}
\newcommand{\Ltg}[2]{\cL^p_{#1,#2}}
\newcommand{\LuT}[1][u]{\cL^p_{#1,T}}
\newcommand{\Ltsp}[1][\infty]{\cL^{#1}_{t,s}}
\newcommand{\Ltsbar}{\overline{\cL}^p_{t,s}}
\newcommand{\Ltsbarp}[1][1]{\overline{\cL}^{#1}_{t,s}}
\newcommand{\tn}{|\!|\!|}

\newcommand{\NAS}{$\mathrm{(}\mathrm{NA}^{\mathrm{s}}\mathrm{)}$}
\newcommand{\NAT}{$\mathrm{(NA2)}$}
\newcommand{\NRA}[1][]{$\mathrm{(NRA)}_{#1}$}
\newcommand{\SNR}[1][]{$\mathrm{(SNR)}_{#1}$}
\newcommand{\NARA}[1][]{$\mathrm{(NARA)}_{#1}$}
\newcommand{\NRAT}{$\mathrm{(NRA2)}$}
\newcommand{\SNA}{$\mathrm{(SNA)}$}
\newcommand{\SNAR}{$\mathrm{(SNAR)}$}
\renewcommand{\subseteq}{\subset}
\renewcommand{\supseteq}{\supset}

\newcommand{\risk}{r}
\newcommand{\vecrisk}{\mathsf{r}}
\newcommand{\rhos}{\mathsf{R}}
\newcommand{\rhonot}{\mathsf{R}^0}

\newcommand{\bX}{{\mbf X}}
\newcommand{\bY}{{\mbf Y}}

\newcommand{\bK}{{\mbf K}}
\newcommand{\bM}{{\mbf M}}
\renewcommand{\Gamma}{\bX}
\let\Xi\varXi
\renewcommand{\epsilon}{\varepsilon}

\renewcommand{\bX}{X}
\renewcommand{\bY}{Y}

\renewcommand{\bK}{K}
\renewcommand{\bM}{M}

\newcommand{\Gr}{\mathrm{Gr}}
\newcommand{\cl}{{\mathrm{cl}}}

\newcommand{\cm}{\mathbf{m}}
\newcommand{\CM}{\mathbf{M}}
\newcommand{\Int}{\mathrm{int}}
\newcommand{\esssup}[1][\cH]{\mathrm{ess\,sup}_{#1}}
\newcommand{\essinf}[1][\cH]{\mathrm{ess\,inf}_{#1}}

\newcommand{\one}{{\mathbf{1}}}
\newcommand{\eqdef}{=}

\usepackage{colortbl}

\newlength{\querylen}
\usepackage{fancybox}

\begin{document}


\title{Risk Arbitrage and Hedging to Acceptability
  under Transaction Costs}

\titlerunning{Risk arbitrage}

\author{Emmanuel Lepinette \and Ilya Molchanov}

\institute{E. Lepinette \at
              Paris-Dauphine University, Place du Mar\'echal De Lattre De
  Tassigny,
  75775 Paris cedex 16, France, and 
  GOSAEF, Tunis-El Manar University, 2092-ElManar, Tunisia\\
  \email{emmanuel.lepinette@ceremade.dauphine.fr}           
  \and I. Molchanov \at Institute of Mathematical Statistics and
  Actuarial Science, University of Bern,
  Alpeneggstr. 22,
  3012 Bern,
  Switzerland\\
  \email{ilya.molchanov@stat.unibe.ch}
}

\date{\today}

\maketitle

\begin{abstract}
  The classical discrete time model of proportional transaction costs
  relies on the assumption that a feasible portfolio process has
  solvent increments at each step. We extend this setting in two
  directions, allowing for convex transaction costs and assuming that
  increments of the portfolio process belong to the sum of a
  solvency set and a family of multivariate acceptable positions,
  e.g. with respect to a dynamic risk measure.  We describe the sets
  of superhedging prices, formulate several no (risk) arbitrage
  conditions and explore connections between them. In the special case
  when multivariate positions are converted into a single fixed asset,
  our framework turns into the no good deals setting. However, in
  general, the possibilities of assessing the risk with respect to any
  asset or a basket of the assets lead to a decrease of superhedging
  prices and the no arbitrage conditions become stronger.
  The mathematical technique relies on results for unbounded and
  possibly non-closed random sets in Euclidean space.
  \keywords{acceptance set \and risk arbitrage \and risk measure \and
    superhedging \and good deal \and solvency set \and random set \and
    transaction costs}
  \subclass{91G20, 60D05, 60G42} 
\end{abstract}

\section{Introduction}
\label{sec:introduction}

Transaction costs in financial markets are often described using
solvency sets, which consist of all
financial positions (in physical quantities) regarded better than the
zero position or at least equivalent to it. 
In the dynamic discrete time setting, the solvency sets form a
set-valued random process $(\bK_t)_{t=0,\dots,T}$ adapted to the
underlying filtration $(\cF_t)_{t=0,\dots,T}$. The no arbitrage
conditions are usually formulated in terms of selections of these
solvency sets, that is, for random vectors that a.s. belong to the solvency
sets and so correspond to particular choices of solvent portfolios. In many
cases, solvency sets are polyhedral cones and the corresponding model is known as
Kabanov's model with proportional transaction costs, see
\cite{kab:ras:str02,kab:saf09,schach01}, where the no arbitrage
conditions are thoroughly discussed. 

If $\xi$ is a claim that matures at time $T$, then the set of initial
positions suitable as a starting value for a self-financing portfolio
process $(V_t)_{t=0,\dots,T}$ paying $\xi$ at maturity forms the
family of superhedging prices for $\xi$.  In the multivariate setting,
the starting values are vectors which are not necessarily comparable
to each other, and so, instead of comparing them by a single numerical
quantity, it is sensible to look at the whole set of superhedging
prices.
The self-financing requirement amounts to the fact that the increment
$V_{t-1}-V_t$ of the portfolio process is solvent at all times, that
is, it a.s.\ belongs to $\bK_t$ for all $t$ (in other words, the increment is
a \emph{selection} of $\bK_t$). 

In order to reduce these superhedging prices, it is possible to
require that the shortfall of the terminal value of the portfolio, in
comparison with the claim, is acceptable with respect to a certain
risk measure. This approach may provide arbitrage opportunities as
\emph{Good Deals}, i.e. terminal claims attainable from the zero
capital and such that the risk of the claim is strictly negative,
equivalently, the utility is strictly positive.  The No Good Deal
condition, first introduced in \cite{coc:saa-00} and then formalised
in \cite{cer:hod01,cher07}, requires that this situation is
impossible.  Unlike the univariate setting, the existence of a good
deal in the multivariate setting does not necessarily mean the
existence of a claim whose multivariate utility belongs to
$\R_+^d$. Indeed, a vector-valued financial position may be acceptable
if some acceptable components compensate for the non-acceptable ones.
This may result in various types of arbitrage opportunities.

Indeed, it is possible to strengthen the no arbitrage requirement by
also considering hedging strategies, where the self-financing
condition is replaced by the acceptability of all intermediate
portfolio changes with respect to a dynamic risk measure, see
\cite{cher07fs}. The setting of \cite{cher07fs,cher08} involves at
least two assets exchangeable without transaction costs and pinpoints
a particular asset that is used as the cash equivalent. The portfolio
is converted to its cash equivalent, with the acceptability condition
imposed on the increments of these cash values for consecutive time
moments. The idea of converting portfolios to a single numerical
quantity with acceptable increments in view of superhedging
one-dimensional claims has been further explored in
\cite{cher:kup:tan17}.

However, if there are several currencies (exchangeable with random
frictionless rates or with transaction costs), it may well be the case
that the position expressed in one currency is acceptable, while the
position in the other one is not, see \cite[Ex.~1.1]{cas:mol14}. This may lead to regulatory arbitrages, see \cite{wil17}. If the regulator is prepared to apply a relaxed acceptability criterion for one currency, then it would be logical to expect the same policy
with respect to another currency or a basket of currencies.  We show
how to handle this in a way that treats all components of a portfolio
in the same manner.\smallskip

The key idea of this work is to extend the family of self-financing
portfolio processes by requiring that $V_{t-1}-V_t$ equals the sum of
a selection of $\bK_t$ (a solvent position) and another random vector
that is not necessarily solvent, but is acceptable with respect to a
dynamic multivariate risk measure. It is worth mentioning that $\bK_t$
is only supposed to be convex, contrarily to the classical literature
of linear transaction costs. With this \emph{hedging to acceptability}
approach, all components of the portfolio are treated in the same
way. Then, $(V_t)_{t=0,\dots,T}$ is called an \emph{acceptable
  portfolio process}. For example, the classical superhedging setting
arises if the componentwise conditional essential infimum is chosen as
the risk measure, so that acceptable random vectors necessarily have
all a.s. non-negative components. The hedging to acceptability
substantially increases the choice of possible hedging strategies, but
in some cases may lead to arbitrage.

\begin{example}
  \label{ex:risk-arbitrage}
  Let $\risk$ be any coherent risk measure.  Consider the one
  period zero interest model with two currencies as the assets. Assume
  that the exchange rate $\pi$ (so that $\pi$ units of the second
  asset buy one unit of the first one) at time one is log-normally
  distributed (in the real world) and the exchanges are free from
  transaction cost. 
  Then, the positions $\gamma'=(-a,\pi a)$ and $\gamma''=(a,-\pi a)$
  for $a>0$ are reachable from $(0,0)$ at zero cost. Their
  risks are $(a,a\risk(\pi))$ and $(-a,a\risk(-\pi))$. In order to
  secure the capital reserves for $\gamma'$, the agent has to reserve
  $a$ of the first currency and $a\risk(\pi)$ of the second one (note
  that $\risk(\pi)<0$). If the exchange rate at time zero is $\pi_0$,
  the initial cost expressed in the second currency is
  \begin{displaymath}
    \pi_0 a+ a \risk(\pi)=a(\pi_0+\risk(\pi)).
  \end{displaymath}
  In order to secure $\gamma''$, the initial cost is
  $a(-\pi_0+\risk(-\pi))$. If $\pi_0$ does not belong to the interval
  $[-\risk(\pi),\risk(-\pi)]$, then either $\pi_0+\risk(\pi)<0$ or
  $-\pi_0+\risk(-\pi)<0$, and we let $a$ grow to release infinite
  capital at time zero. Note that this model does not admit financial
  arbitrage, since there exists a martingale measure. 
  This example can be easily modified by accounting for proportional 
  transaction costs.
\end{example}

It is recognised by now that risks of multivariate positions involving
possible exchanges of assets and transaction costs are described as
sets, see \cite{cas:mol07,ham:hey10}. The multiasset setting naturally
makes it possible to offset a risky position using various
combinations of assets. In this framework it is also natural to
consider the family of all attainable positions as a set-valued
portfolio.
Treating both arguments and values of
a risk measure as random sets leads to law invariant risk measures and
makes it possible to iterate the construction, which is essential to
handle dynamic risk measures.

One of the aims of this paper is to introduce a geometric
characterisation of superhedging prices. On the way, we suggest a
constructive definition of dynamic multivariate risks based on the
families of acceptable positions and so extend the existing works on
dynamic set-valued risk measures \cite{fein:rud13,fein:rud14fs} by
letting the arguments of risks and their values be sets of random
vectors in $\R^d$. In many instances, these sets may be interpreted as
random (possibly, non-closed) sets. The necessary background on random
sets is provided in the appendix. In particular, it is shown that the
Minkowski (elementwise) sum of two random closed sets is measurable,
no matter if the sum is closed or not. Special attention is devoted to
the decomposability and infinite decomposability properties, which are
the key concepts suitable to relate families of random vectors to
selections of random sets.

We refer to \cite{delb12} and \cite{foel:sch04} for the basics of
static risk measures and to \cite{acc:pen11} for a survey of the
dynamic $\Lnot[\infty]$-setting, further extended beyond the
$\Lnot[\infty]$-setting by the module
approach worked out in \cite{fil:kup:vog09,fil:kup:vog12}.

Static risk measures are usually defined on $\Lnot[p](\R,\cF)$ with
$p\in[1,\infty]$. However, in many cases, they are well defined also
on larger sets of random variables. For example,
$\risk(\xi)=-\essinf[]\xi$ makes sense for all random variables
essentially bounded from below by a constant. Similarly, if
$\risk(\xi)=-\E\xi$, then the acceptance set is defined as the family
of $\xi$ such that their positive and negative parts satisfy
$\E\xi^+\geq\E\xi^->\infty$. The boundedness of $\E\xi^+$ is not
required.

To account for similar effects in relation to multivariate dynamic
risk measures, we put forward acceptance sets in place of risk
measures.  The \emph{acceptance sets} $\Acc_{t,s}$ with $t\leq s$ are
subsets of the sum of the family of $\cF_s$-measurable random vectors
in $\R^d$ that admit generalised conditional $p$th moment with respect
to $\cF_t$ and the family of all $\cF_s$-measurable random vectors in
$\R_+^d$.  Section~\ref{sec:select-risk-meas} introduces basic
conditions on the acceptance sets and several optional ones.

The dynamic \emph{selection risk measure} $\rhos_{t,s}(\Xi)$ for a
family $\Xi\subset\Lnot(\R^d,\cF_T)$ is introduced as the
closure in probability of $(\Xi+\Acc_{t,s})\cap\Lnot(\R^d,\cF_t)$. 
If $\Xi=\Lnot(\bX,\cF_T)$ is the family of selections for a random
closed set $\bX$, then $\rhos_{t,s}(\bX)=\rhos_{t,s}(\Xi)$ itself is
an $\cF_t$-measurable random closed set. In comparison with
\cite{fein:rud14fs}, this approach explicitly defines a set-valued
dynamic risk measure instead of imposing on it some axiomatic
properties. This yields a set-valued risk measure with a set-valued
argument that can be naturally iterated in the dynamic framework. The
conditional convexity of the acceptance sets yields that
\begin{displaymath}
  \rhos_{t,s}(\lambda\bX+(1-\lambda)\bY)\supseteq
  \lambda\rhos_{t,s}(\bX)+(1-\lambda)\rhos_{t,s}(\bY)
  \quad \text{a.s.}
\end{displaymath}
for any $\lambda\in\Lnot([0,1],\cF_t)$ and random closed sets $\bX$
and $\bY$, meaning that the risk measure is also conditionally convex.  The
static case of this construction was considered in \cite{cas:mol14},
where properties of selection risk measures in the coherent case are
obtained, some of them easily extendable for the dynamic convex
setting. Comparing to \cite{cas:mol14}, we work with solvency sets
instead of portfolios available at price zero and also allow the
argument of the risk measure to be a rather general family of random
vectors.  

The hedging to acceptability relies on a sequence
$(\bK_t)_{t=0,\dots,T}$ of solvency sets and the acceptance sets
$\Acc_{t,s}$ for $0\leq t\leq s\leq T$. Note that the solvency sets
are not assumed to be conical, since non-conical models naturally
appear, e.g. in the order book setting, see \cite{cet:jar:prot04} and
\cite{pen:pen10}.  An acceptable portfolio process
$(V_t)_{t=0,\dots,T}$ introduced in Section~\ref{sec:hedg-with-coher}
satisfies $V_{t-1}-V_t=k_t+\eta_t$ for $k_t\in\Lnot(\bK_t,\cF_t)$,
$\eta_t\in\Acc_{t-1,t}$, and all $t$. In other words, the available
assets do suffice to pay for the portfolio at the next step up to an
amount acceptable with respect to some risk measure. The strongest
acceptability condition assumes that $\Acc_{t-1,t}$ consists of random
vectors with non-negative components and yields the classical
arbitrage theory for markets with transaction costs, see
\cite{kab:saf09}. The weakest acceptability requirement presumes that
all $\eta_t$ from $\Acc_{t-1,t}$ have non-negative
$\cF_{t-1}$-conditional expectations.

If $\xi$ is a terminal claim on $d$ assets, then $\Xi_t^\xi$ denotes
the set of all initial endowments at time $t$ that ensure the
existence of an acceptable portfolio process paying $\xi$ at maturity,
that is, $V_T\in\xi+\bK_T$ a.s. Equivalently, $\Xi_t^\xi$ is the
family of $\cF_t$-measurable elements of $(\xi-\Att_{t,T})$, where
$\Att_{t,T}$ is the set of claims attainable at time $T$ starting from
zero investment at time $t$. The family $\Xi_t^\xi$ may be used to
assess the risk associated with $\xi$ at time $t$.

The no arbitrage conditions we study are imposed on the set of
super-hedging prices $\Xi_t^0$ for the zero claim $\xi=0$. They may be
reformulated as no arbitrage conditions on the set of attainable
claims, which are weaker than the usual ones of the literature.  These
\emph{no risk arbitrage} conditions are introduced and analysed in
Section~\ref{sec:risk-arbitrage}. In difference to \cite{cher07fs}, we
do not rely on the weak compactness of the duals to the acceptance
sets and we do not need to pinpoint any reference asset.  It should be
noted that the risk arbitrage only makes sense in the multiasset
setting with some trading opportunities between the assets; if
$\bK_t=\R_+^d$ (which is always the case on the line), then all no
risk arbitrage conditions automatically hold.

It is shown that, in some cases, it is possible to represent the
families of capital requirements as a set-valued process, and the no
risk arbitrage conditions, for linear transaction costs, can be characterised in terms of weakly
consistent price systems, so providing a variant of the Fundamental
Theorem of Asset Pricing in our framework, see
Theorem~\ref{coroEquivPropNARA}. 

A comparison of our approach with the no good deals setting is
provided in Section~\ref{sec:no-good-deals-1}. It is shown that our
approach imposes stronger no arbitrage conditions that are more
difficult to check, but which result in lower superhedging prices.

Note that the sets $\Acc_{t,s}$ of acceptable positions always contain the
family $\Lnot(\R_+^d,\cF_s)$ of random vectors with a.s. non-negative
components and, in many cases, $\Acc_{t,s}$ is a subset of the family
of random vectors with non-negative generalised conditional
expectation given $\cF_t$.  Thus, the no risk arbitrage conditions are
sandwiched between those for the risk measure based on the conditional
essential infimum and on the conditional expectation. The first choice
corresponds to the classical financial arbitrage with transaction
costs, where our no risk arbitrage conditions become the classical
no arbitrage conditions. 

Section~\ref{sec:conditional-core-as} recovers and extends several
results from \cite{kab:saf09}. In this classical setting, our approach
yields a new geometric interpretation of the sets of
superhedging prices with possibly non-conical solvency sets; it is
formulated using the
concept of the conditional core of a random set elaborated in
\cite{lep:mol19}. The result applies also in some cases when the
classical consistent price systems characterisation
fails.  

In Section~\ref{sec:sandwich-theorems} we characterise no arbitrage
conditions arising by adopting the generalised conditional expectation
as the acceptability criterion. These are the strongest no arbitrage
conditions in our framework; their validity ensures the absence of 
arbitrage for all acceptance criteria satisfying a dynamic version of
the dilatation monotonicity condition from
\cite{cher:grig07}. 

The results from Sections~\ref{sec:conditional-core-as} and
\ref{sec:sandwich-theorems} are illustrated on a two-asset example in
Section~\ref{sec:appl-two-dimens}.

\section{Dynamic acceptance sets and selection risk measures}
\label{sec:select-risk-meas}

\subsection{Definition and main properties}
\label{sec:cond-risk-meas}

Let $(\Omega,(\cF_t)_{t=0,\dots,T},\P)$ be a stochastic basis on a
complete probability space such that $\cF_0$ is the trivial
$\sigma$-algebra and $\cF_T=\cF$. In the following,
we endow random vectors and events with a subscript that indicates the
$\sigma$-algebras they are measurable with respect to. The subscript
is often omitted for $\cF_T$-measurable random vectors.  

Let $\Lnot[p](\R^d,\cF)$ with $p\in[1,\infty]$ be the family of
$p$-integrable random vectors (essentially bounded if $p=\infty$), and
let $\Lnot(\R^d,\cF)$ be the family of all random vectors in $\R^d$.
The closure in the strong topology in
$\Lnot[p](\R^d,\cF)$ for $p\in[1,\infty)$ is denoted by $\cl_p$, and
$\cl_0$ is the closure in probability in $\Lnot(\R^d,\cF)$. If $p=\infty$, the
closure is considered with respect to the a.s. convergence of
uniformly bounded sequences.

For a sub-$\sigma$-algebra $\cH\subset\cF$, denote $\Lnot[p_{\cH}](\R^d,\cF)$
the module of $\cF$-measurable random vectors that can be represented
as $\gamma\xi$ with $\xi\in\Lnot[p](\R^d,\cF)$ and
$\gamma\in\Lnot(\R,\cH)$, see \cite[Ex.~2.5]{fil:kup:vog09}. In
particular, $\Lnot[1_{\cH}](\R^d,\cF)$ is the family of all $\xi$ that
admit the generalised conditional expectation $\E^g(\xi|\cH)$ with
respect to $\cH$, see \cite[Lemma~B.3]{lep:mol19}. Following
\cite[Ex.~2.5]{fil:kup:vog09}, the module norm is defined by
\begin{equation}
  \label{eq:norm-module}
  \tn\xi\tn_{p,\cH}=
  \begin{cases}
    \E(\|\xi\|^p |\cH)^{1/p}, & p\in[1,\infty),\\
    \esssup[\cH]\|\xi\|, & p=\infty,
  \end{cases}
\end{equation}
where $\esssup[\cH]\|\xi\|$ is the $\cH$-measurable essential supremum
of $\|\xi\|$, see \cite[Appendix~A.5]{foel:sch04}.  We endow the space
$\Lnot[p_{\cH}](\R^d,\cF)$ with the topology of
$\Lnot[p_{\cH}]$-convergence by assuming that $\xi^n$ converges to
$\xi$ if $\tn\xi^n-\xi\tn_{p,\cH}\to0$ in probability if
$p\in[1,\infty)$. If $p=\infty$, we use the $\cF_t$-bounded
convergence in probability, meaning that $\esssup[\cH]\|\xi^n\|$ is
bounded and $\tn(\xi^n-\xi)\wedge 1\tn_{1,\cH}\to0$ in probability as
$n\to\infty$. 

Denote shortly $\Lnot[p]=\Lnot[p](\R^d,\cF)$,
$\Lts=\Lnot[p_{\cF_t}](\R^d,\cF_s)$ for $t\leq s$, and let
$\Ltsbar=\barLnot[p_{\cF_t}](\R^d,\cF_s)$ be the family of random
vectors $\xi_s$ that can be decomposed as $\xi_s=\xi'_s+\xi''_s$,
where $\xi'_s\in \Lts$ and $\xi_s''\in\Lnot(\R_+^d,\cF_s)$.  Following
the classical definition of an acceptance set in the theory of risk
measures, we introduce the acceptance set $\Acc_{t,s}$ for $t\le s$ as
the collection of all $\cF_s$-measurable financial positions regarded
acceptable at time $t$.

\begin{definition}
  \label{DefiRiskmeasure}
  Discrete time \emph{$\Lnot[p]$-dynamic convex acceptance sets} are
  a family $\{\Acc_{t,s},\; 0\leq t\leq s\leq T\}$, such that
  $\Acc_{t,s}\subset\Ltsbar$
  and the following properties hold for all $0\leq t\leq s\leq T$. 
  \begin{enumerate}[(i)]
  \item Normalisation: $\Acc_{t,t}=\Lnot(\R_+^d,\cF_t)$,
    $\Acc_{t,s}\supset\Lnot(\R_+^d,\cF_s)$, and\\
    $\Acc_{t,s}\cap\Lnot(\R_-^d,\cF_t)=\{0\}$. 
  \item Integrability: 
    \begin{equation}
      \label{eq:acc-integrability}
      \Acc_{t,s}=(\Acc_{t,s}\cap \Lts)
      +\Lnot(\R_+^d,\cF_s).
    \end{equation}
  \item Closedness: $\Acc_{t,s}\cap \Lts[T]$ is 
    closed in $\Lts[T]$. 
  \item Conditional convexity: for all $\alpha_t\in
    \Lnot([0,1],\cF_t)$, and $\eta'_s,\eta''_s\in\Acc_{t,s}$,
    \begin{displaymath}
      \alpha_t \eta'_s +(1-\alpha_t)\eta''_s \in \Acc_{t,s}.
    \end{displaymath}
  \item Weak time consistency: $\Acc_{t,s}\cap
    \Lnot(\R^d,\cF_u)=\Acc_{t,u}$ for all $0\leq t\leq u\leq s\leq T$.
  \item Compensation: if $\xi_s\in\Lts$, then
    $(\xi_s+\Acc_{t,s})\cap \Lnot(\R^d,\cF_t)\neq\emptyset$.
  \end{enumerate}
\end{definition}

The integrability property implies that $\Acc_{t,s}$ is an \emph{upper
  set}, that is, $\eta_s\in\Acc_{t,s}$ and $\eta_s\leq \eta'_s$
a.s. (all inequalities are understood coordinatewisely) yield
$\eta'_s\in\Acc_{t,s}$.  The compensation property implies that, for
all $\xi_s\in\Lts$, there exists $\gamma_t\in \Lnot(\R^d,\cF_t)$ such
that $\gamma_t+\xi_s\in \Acc_{t,s}$, i.e. it is possible to make the
financial position $\xi_s$ acceptable by adding the position
$\gamma_t$.  In the following, we need only the acceptance sets
$\Acc_{t-1,t}$ for $t=1,\dots,T$.

\begin{example}[Static univariate convex risk measures]
  \label{ex:static-convex}
  Consider the one-period setting in one dimension with $t=0,1$. If
  $\risk$ is a convex $\Lnot[p]$-risk measure with $p\in[1,\infty)$,
  then its acceptance set $\Acc_{0,1}\cap \Lnot[p](\R,\cF_1)$ is the
  family of $\eta_1\in\Lnot[p](\R,\cF_1)$ such that $\risk(\eta_1)\leq
  0$.  The lower semicontinuity of $\risk$ is equivalent to the
  closedness of the acceptance set. The conditional convexity property
  of the acceptance set is equivalent to the convexity property of the
  risk measure. The compensation property corresponds to the
  finiteness of $\risk$.
\end{example}

The following result refers to the infinite decomposability property
(see Definition~\ref{def:decomp}), also known as the countable
concatenation property \cite{fil:kup:vog09} or $\sigma$-stability
\cite{fil:kup:vog12}.

\begin{lemma} 
  \label{ACC-inf-decomp} 
  The families $\Acc_{t,s}$ are infinitely $\cF_t$-decomposable for
  all $0\leq t\leq s\leq T$.
\end{lemma}
\begin{proof} 
  If $\eta_s^i\in \Acc_{t,s}\cap
  \Lts$ and $B_t^i\in \cF_t$, $i\ge 1$, then
  \begin{displaymath}
    \bar\eta^n_s\eqdef\sum_{i=1}^{n}\one_{B_t^i}\eta_s^i
    +\eta_s^1\one_{\Omega\setminus\cup_{i=1}^n B_t^i}\in \Acc_{t,s}
  \end{displaymath}
  by the conditional convexity property, so that $\Acc_{t,s}\cap\Lts$
  is $\cF_t$-decomposable. Since
  $\bar\eta^n_s\to\bar\eta_s\eqdef\sum_{i=1}^{\infty}\one_{B_t^i}\eta_s^i$
  in the $\tn\cdot\tn_{p,\cF_t}$-norm if $p\in[1,\infty)$ and
  $\cF_t$-boundedly in probability if $p=\infty$, we have
  $\bar\eta_s\in\Acc_{t,s}\cap\Lts$. By the integrability property,
  $\Acc_{t,s}$ is also infinitely decomposable.
\end{proof}

\begin{definition}
  \label{def:extra}
  The family of dynamic convex acceptance sets is called
  \begin{enumerate}[(i)]
  \item \emph{coherent} if $\alpha_t\eta_s\in\Acc_{t,s}$ for all
    $t\leq s$, $\alpha_t\in\Lnot([0,\infty),\cF_t)$, and
    $\eta_s\in\Acc_{t,s}$;
  \item \emph{continuous from below at zero} if, for all $t\leq s$,
    and any sequence $\xi_s^n\in \Lnot[p_{\cF_t}](\R_-^d,\cF_s)$, 
    $n\geq1$, with
    $\tn \xi_s^n\tn_{p,\cF_t}\to0$ in probability as $n\to\infty$,
    there exists a sequence $\gamma_t^n\in\Lnot(\R_+^d,\cF_t)$,
    $n\geq1$ and $k\in\R_+$, such that
    $\gamma_t^n+\xi^n_s\in \Acc_{t,s}$ and $\|\gamma_t^n\|\leq k
    \tn\xi_s^n\tn_{p,\cF_t}$ a.s. for all $n$.
  \end{enumerate}
\end{definition}

If $p=\infty$ then the continuity from below always holds and is
easily verified by choosing $\gamma_t^n$ with all identical components
being $\tn\xi_s^n\tn_{\infty,\cF_t}$.

\begin{example}
  The acceptance sets can be defined using a univariate convex dynamic
  $\Lnot[p]$-risk measure $(\risk_t)_{t=0,\dots,T}$, so that
  $\Acc_{t,s}\cap\Lts$ is the $d$th Cartesian power of the acceptance set
  for $\risk_t$. Equivalently, 
  \begin{math}
    \xi_s 
    \in \Acc_{t,s}\cap\Lts
  \end{math}
  if and only if all components of $\xi_s$ are acceptable under
  $\risk_t$. 
  The continuity from below property (with $p\in[1,\infty)$) holds if
  $\risk_t$ is lower semicontinuous in the
  $\tn\cdot\tn_{p,\cF_t}$-norm and continuous from below, which is the
  case if
  $\risk_t$ is convex and a.s. finite, see \cite[Th.~4.1.4]{vog09}.
\end{example}

\begin{example}[Dual construction of conditional acceptance sets]
  \label{dyn-ex-2} 
  Let $p=\infty$, and consider families $\cZ_{t,s}\subseteq
  \Lnot[1_{\cF_t}](\R^d_+,\cF_s)$ with $0\leq t\le s\le T$ such that
  $\cZ_{t,u}\subseteq \cZ_{t,s}$ for all $t\le u\le s$, and 
  $\E^g(\zeta_s|\cF_t)=(1,\dots,1)$ for all $\zeta_s\in \cZ_{t,s}$.
  Note that we do not assume that $\cZ_{t,s}$ is weakly compact. 
  Define 
  \begin{displaymath}
    \Acc_{t,s}=\Lnot(\R^d_+,\cF_s)+\bigcap_{\zeta_s\in \cZ_{t,s}}
    \big\{\eta_s\in\Ltsp:
    \;\E^g(\langle \zeta_s,\eta_s\rangle|\cF_t)\ge 0\big\},
  \end{displaymath}
  where $\langle \zeta_s,\eta_s\rangle$ is the scalar product. 
  It is easily seen that conditions (i), (ii), (iv) and (v) of
  Definition~\ref{DefiRiskmeasure} hold and these acceptance sets are
  coherent.

  If $\xi_s^n\to\xi_s$ in probability for $\xi_s^n\in\Acc_{t,s}$, and
  all $\xi_s^n$ are bounded in the norm by
  $\gamma_t\in\Lnot(\R_+,\cF_t)$, then
  $\E^g(\langle \zeta_s,\xi_s\rangle|\cF_t)\geq0$ by the dominated
  convergence theorem for generalised conditional expectations. Thus,
  condition (iii) also hold.

  If $\xi_s\in\Ltsp$, then the components of $\xi_s$ are bounded in
  the absolute value by $\eta_t\in\Lnot(\R_+^d,\cF_t)$. Then,
  $\eta_t-\xi_s$ is non-negative and so belongs to $\Acc_{t,s}$, and
  $\xi_s+(\eta_t-\xi_s)\in \Lnot(\R^d,\cF_t)$.
  Thus, (vi) also holds.   
\end{example}

\subsection{Dynamic selection risk measures}
\label{sec:dynam-select-risk}

Let $\Xi_T$ be an \emph{upper} subset of $\Lnot(\R^d,\cF_T)$, that is,
with each $\xi\in\Xi_T$, the family $\Xi_T$ also contains all
$\xi'\in \Lnot(\xi+\R_+^d,\cF_T)$. The most important example of
such family is the family of selections $\Lnot(\bX_T,\cF_T)$ for an
$\cF_T$-measurable \footnote{The (graph) measurability of a random set is defined in the appendix.} upper random set $\bX_T$ in $\R^d$, that is,
$\bX_T+\R_+^d\subset\bX_T$ a.s. If $\bX_T$ is also closed, then its
centrally symmetric version $(-\bX_T)$ is a \emph{set-valued
  portfolio} in the terminology of \cite{cas:mol14}.

\begin{definition}
  \label{def:sel-risk}
  Let $\Xi_T\subset\Lnot(\R^d,\cF_T)$ be an upper set.
  For $t\leq s\leq T$,
  \begin{equation}
    \label{eq:union-risk}
    \rhonot_{t,s}(\Xi_T)=(\Xi_T+\Acc_{t,s})\cap \Lnot(\R^d,\cF_t) 
  \end{equation}
  denotes the set of all $\gamma_t\in \Lnot(\R^d,\cF_t)$, such that
  $\gamma_t-\xi\in\Acc_{t,s}$ for some $\xi\in\Xi_T$. Let
  $\rhos_{t,s}(\Xi_T)$ denote the closure in probability of
  $\rhonot_{t,s}(\Xi_T)$. If $\Xi_T=\Lnot(\bX_T,\cF_T)$ is the family
  of selections of an upper random set $\bX_T$, we write
  $\rhonot_{t,s}(\bX_T)$ and $\rhos_{t,s}(\bX_T)$ instead of
  $\rhonot_{t,s}(\Xi_T)$ and $\rhos_{t,s}(\Xi_T)$. In view of this,
  $\rhos_{t,s}(\bX_T)$ (and also $\rhos_{t,s}(\Xi_T)$) is called the
  \emph{dynamic selection risk measure}.
\end{definition}

Note that $\rhonot_{T,T}(\Xi_T)=\Xi_T$,
$\rhonot_{t,s}(\Xi_T)=\rhonot_{t,s}(\Xi_T\cap\Lnot(\R^d,\cF_s))$, and
$\rhonot_{t,u}(\Xi_T)\subset\rhonot_{t,s}(\Xi_T)$ for all $0\leq t\leq
u\leq s\leq T$.  If only portfolios from a random set $\bM_t$ are
allowed for compensation at time $t$, as it is the case in
\cite{fein:rud14fs}, it is easy to modify \eqref{eq:union-risk} by
intersecting $(\Xi_T+\Acc_{t,s})$ with
$\Lnot(\bM_t,\cF_t)$. 

The empty selection risk measure corresponds to completely
unacceptable positions.  The compensation property of acceptance sets
guarantees that $\rhonot_{t,s}(\Xi_T)$ is not empty if
$\Xi_T\cap\Lts\neq\emptyset$.  The family $\Xi_T$ is said to be
\emph{acceptable} for the time horizon $s$ if
$0\in\rhonot_{t,s}(\Xi_T)$, equivalently, $-\Xi_T$ contains an element from
$\Acc_{t,s}$.
The dynamic selection risk measure is conditionally convex, that is,
\begin{displaymath}
  \rhonot_{t,s}(\alpha_t\Xi'_T+(1-\alpha_t)\Xi''_T)
  \supseteq \alpha_t\rhonot_{t,s}(\Xi'_T)+(1-\alpha_t)\rhonot_{t,s}(\Xi''_T)
\end{displaymath}
for all $\alpha_t\in\Lnot([0,1],\cF_t)$, and the same holds for
the closures.  The next result follows from Lemma~\ref{ACC-inf-decomp}.

\begin{lemma}
  \label{lemma:xi-idec}
  If $\Xi_T$ is infinitely $\cF_t$-decomposable, then $\rhonot_{t,s}(\Xi_T)$ and
  $\rhos_{t,s}(\Xi_T)$ are also infinitely $\cF_t$-decomposable. 
\end{lemma}

\begin{lemma}
  \label{lemma:det-set}
  Let $\bX_T$ be an $\cF_T$-measurable random upper closed set.
  \begin{enumerate}[(i)]
  \item $\rhos_{t,s}(\bX_T)$ is the family of selections of an
    $\cF_t$-measurable random upper closed set in $\R^d$, which is
    denoted by $\rhos_{t,s}(\bX_T)$.
  \item If $\bX_T$ is a.s. convex, then $\rhonot_{t,s}(\bX_T)$ is
    a.s. convex. If $\bX_T$ is a cone and the acceptance sets are
    coherent, then $\rhonot_{t,s}(\bX_T)$ is a cone.
  \end{enumerate}
\end{lemma}
\begin{proof}
  (i) By Lemma~\ref{lemma:xi-idec}, $\rhos_{t,s}(\bX_T)$ is an
  $\cF_t$-decomposable family, and so Theorem~\ref{ReprDecompL0}
  applies.

  (ii) If $\gamma_t^1,\gamma_t^2\in\rhonot_{t,s}(\bX_T)$, then
  $\gamma_t^i-\xi_s^i\in\Acc_{t,s}$, $i=1,2$, for
  $\xi_s^1,\xi_s^2\in\Lnot(\bX_T,\cF_s)$. For any $t\in(0,1)$,
  the conditional convexity property yields that
  $t\gamma_t^1+(1-t)\gamma_t^2-\xi_s\in\Acc_{t,s}$ with
  $\xi_s=t\xi_s^1+(1-t)\xi_s^2\in \Lnot(\bX_T,\cF_s)$. The conical
  property is trivial. 
\end{proof}

\begin{example}
  \label{ex:univariate-simple}
  If $\bX_s=\xi_s+\R_+^d$ for $\xi_s\in\Lts$, then
  $\rhonot_{t,s}(\bX_s)=\rhos_{t,s}(\bX_s)=\vecrisk_t(-\xi_s)+\R_+^d$ for a
  vector-valued dynamic risk measure $\vecrisk_t$ on $\Lts$, see \cite{vog09}.
\end{example}

\section{Hedging to acceptability}
\label{sec:hedg-with-coher}

\subsection{Acceptable portfolio process}
\label{sec:accept-portf-proc}

Let $(\bK_t)_{t=0,\dots,T}$ be a sequence of random closed convex
sets, such that, for all $t$, we have $\bK_t\cap\R_-^d=\{0\}$, $\bK_t$ is an
upper set, and $\bK_t$ is $\cF_t$-measurable. The set $\bK_t$ is
understood as the family of all solvent positions at time $t$
expressed in physical units and is called a \emph{solvency set}, see
\cite{kab:saf09}. If the solvency sets are cones, this model is well
studied under the name of Kabanov's model; it describes the market
subject to proportional transaction costs, see
\cite{kab:saf09,schach01}. If the solvency sets are cones and the
acceptance sets are coherent, we talk about the \emph{coherent
  conical} setting.

Let $\bK_t^0$ be the largest $\cF_t$-measurable linear subspace
contained in $\bK_t$, that is,
\begin{displaymath}
  \bK_t^0=\bigcap_{c\ne 0}c \bK_t
  =\bigcap_{c\in \QQ\setminus\{0\}}c \bK_t,
\end{displaymath}
which is also a random closed set.  The solvency sets are said to be
\emph{proper} if $\bK_t^0=\{0\}$ and \emph{strictly proper} if
$\tilde{\bK}_t\eqdef\bK_t\cap(-\bK_t)=\{0\}$ for all $t=0,\dots,T$.
If $\bK_t$ is a cone, then $\tilde{\bK}_t=\bK_t^0$, while in general
$\bK_t^0\subset\tilde{\bK}_t$. Since $\tilde\bK_t$ is convex and
origin symmetric, $\bK_t$ is proper if and only if $\tilde\bK_t$ is
bounded.

\begin{definition}
  \label{Portfolio-process} 
  A sequence $V_t\in\Lnot(\R^d,\cF_t)$, $t=0,\dots,T$,
  is called an \emph{acceptable} portfolio process if
  \begin{equation}
    \label{eq:v-acc}
    V_{t-1}-V_t\in
    \Lnot(\bK_t,\cF_t)+\Acc_{t-1,t},\qquad t=1,\dots,T.
  \end{equation}
\end{definition}

By the definition of the selection risk measure, \eqref{eq:v-acc} is
equivalent to
\begin{equation}
  \label{eq:12}
  V_{t-1} \in \rhonot_{t-1,t}(V_t+\bK_t),\qquad t=1,\dots,T.
\end{equation}
Thus, paying transaction costs, it is possible to transform
$V_{t-1}-V_t$ into an acceptable position for the horizon
$t$. Equivalently, $V_{t-1}$ does suffices to purchase
$V_t+k_t+\eta_t$ for some $k_t\in\bK_t$ and
$\eta_t\in\Acc_{t-1,t}$. 
The \emph{initial endowment} at time $t$ is any $V_{t-}\in
\Lnot(V_t+\bK_{t},\cF_{t})$, so that it is possible to convert
$V_{t-}$ into $V_t$ paying the transaction costs.

\subsection{Attainable positions and superhedging}
\label{sec:atta-posit-superh}

The family of \emph{attainable positions} at time $s>t$ is the family
of random vectors that may be obtained as $V_s$ for acceptable
portfolio processes starting from zero investment at time $t$. By
\eqref{eq:v-acc}, the family of attainable positions is given by
\begin{displaymath}
  \Att_{t,s}
  = \sum_{u=t}^s \Lnot(-\bK_u,\cF_u)-\sum_{u=t}^{s-1} \Acc_{u,u+1}.
\end{displaymath}

Let $\xi\in\Lnot(\R^d,\cF_T)$ be a \emph{terminal claim} (or payoff).
The hedging to acceptability aims to come up
with an acceptable portfolio process $(V_t)_{t=0,\dots,T}$ that
guarantees paying $\xi$ in the sense that the terminal wealth $V_T$
belongs to $\Xi_T^{\xi}\eqdef\Lnot(\bX_T^\xi,\cF_T)$, being the family of
selections of the random closed set $\bX_T^\xi\eqdef\xi+\bK_T$. Define
recursively
\begin{equation}
  \label{RelPortfolio}
  \Xi_{t}^{\xi}\eqdef\Lnot(\bK_{t},\cF_{t})+\rhonot_{t,t+1}(\Xi_{t+1}^{\xi}),
  \qquad t=T-1,\dots,0,
\end{equation}
which is the set of time $t$ superhedging prices for the claim $\xi$.
The family
$\Xi_t^\xi$ consists of the time $t$ superhedging prices for $\xi$ and
so may serve as a dynamic convex risk measure of $\xi$ with values
being subsets of $\Lnot(\R^d,\cF_t)$.
If $\xi=\xi'-\xi''$ for $\xi'\in\Lnot[p](\R^d,\cF_T)$ and
$\xi''\in\Lnot(\R_+^d,\cF_T)$, the compensation property of acceptance
sets ensures that $\Xi_t^\xi$ is not empty for all $t$. 
 
In order to handle the asymptotic version of the risk measure, let
$\hat\Xi_T^{\xi}\eqdef\Xi_T^\xi$, and further
\begin{equation}
  \label{RelPortfolioAsymptotic}
  \hat\Xi_{t}^{\xi}\eqdef\Lnot(\bK_{t},\cF_{t})
  +\rhos_{t,t+1}( \Xi_{t+1}^{\xi}),
  \qquad t=T-1,\dots,0.
\end{equation} 
Note that $\Xi_t^{\xi}\subseteq
\hat\Xi_t^{\xi}\subseteq \cl_0(\Xi_t^{\xi})$, whence
$\cl_0(\hat\Xi_t^{\xi})=\cl_0(\Xi_t^{\xi})$ for all $t$. The families
$\Xi_{t}^{0}$ and $\hat\Xi_T^{0}$ arise by letting $\xi=0$ a.s. 

\begin{lemma}
  \label{lemma:st-price}
  \begin{enumerate}[(i)]
  \item The families $\rhonot_{t,s}(\Xi_s^{\xi})$,
    $\rhos_{t,s}(\Xi_s^{\xi})$, and $\hat\Xi_t^{\xi}$ are convex and infinitely
    $\cF_t$-decomposable for all $0\leq t\le s\leq T$.
  \item For each $t\le T$, there exists a (possibly non-closed) random
    set $\bX_t^{\xi}$ such that
    $\hat\Xi_t^{\xi}=\Lnot(\bX_t^{\xi},\cF_t)$.
  \item For any $t\leq T$, the family of all initial endowments
    $V_{t-}$ at time $t$, allowing to start an acceptable portfolio
    process $(V_s)_{t\le s\leq T}$ such that
    $V_T\in\Lnot(\xi+\bK_T,\cF_T)$ a.s., coincides with $\Xi_t^\xi$ and
    \begin{equation}
      \label{eq:4}
      \Xi_t^\xi=(-\Att_{t,T}+\xi)\cap \Lnot(\R^d,\cF_t).
    \end{equation}
  \item If $\bK_T$ is a cone, then $\Xi_t^\xi\subset\Xi_t^0$ for any
    $\xi\in\Lnot(\bK_T,\cF_T)$.
  \end{enumerate}
\end{lemma}
\begin{proof}  
  (i) follows from Lemma~\ref{lemma:xi-idec}.

  (ii) The existence of $\bX_t^\xi$ is trivial for $t=T$. Suppose that it
  holds at time $t$. The result for $t-1$ follows from the induction
  assumption and \eqref{RelPortfolioAsymptotic} by
  Lemma~\ref{A+B-measurable}.

  (iii) follows from the fact that $(\gamma_T+\Att_{t,T})\cap
  (\xi+\bK_T)\neq\emptyset$ if and only if $\gamma_T\in(-\Att_{t,T}+\xi)$.
  
  (iv) follows from (iii), since $\xi+\bK_T\subset
  \bK_T$ a.s. 
\end{proof}

\begin{example}
  \label{ex:univariate}
  If $\bK_t=\R_+^d$ a.s. for all $t$ (which is always the case if
  $d=1$), then an acceptable portfolio process satisfies
  $V_{t-1}-V_t\in\Acc_{t-1,t}$ for all $t=1,\dots,T$. Then,
  \begin{displaymath}
    \Xi_t^\xi=(\xi+\sum_{s=t}^{T-1}\Acc_{s,s+1})\cap\Lnot(\R^d,\cF_t).
  \end{displaymath}
  Since
  $\Acc_{t-1,t}\cap\Lnot(\R^d,\cF_{t-1})=\Lnot(\R_+^d,\cF_{t-1})$ for
  all $t\ge 1$ by the weak time consistency and normalisation
  properties, 
  the induction argument yields that
  $\Xi_{t}^0=\Lnot(\R_+^d,\cF_t)$.  If $\xi$ does not a.s.\ vanish,
  the set $\rhonot_{t,s}(\Xi^\xi_T)$ becomes non-trivial. Its static
  variant is called a regulator risk measure in \cite{ham:rud:yan13};
  it only takes into account the acceptability requirement and
  disregards any trading opportunities between the components. In the
  terminology of \cite{ham:rud:yan13}, $\rhos_{0,1}(\xi+\bK_1)$ (in
  the static setting with a conical $\bK_1$) is called the market
  extension of the regulator risk measure.
\end{example}

\section{Risk arbitrage}
\label{sec:risk-arbitrage}

Recall that $\Xi_t^0$ is the set of time $t$ super-hedging prices for
the zero claim. By \eqref{eq:4},
\begin{equation}
  \label{eq:1}
  \Xi_t^0=(-\Att_{t,T})\cap \Lnot(\R^d,\cF_t).
\end{equation}
For multivariate financial models, e.g. models with proportional
transaction costs, several no arbitrage conditions have been
considered. In Kabanov's model, there is the NA condition, its robust
version NA$^r$, but also the NA2 condition derived using an
alternative approach, see \cite{ras08a} and \cite{kab:saf09}. All
these conditions are formulated in terms of the set of all terminal
claims $\Att_{t,T}$ attainable from the zero initial endowment. Here,
we consider weaker no arbitrage conditions imposed on the
super-hedging prices for the zero claim.

\begin{definition} 
  \label{def:nra-multiperiod}
  The multiperiod model satisfies
  \begin{description}
  \item[\SNR] (strict no risk arbitrage) if  
    $\hat\Xi_t^0\cap\Lnot(-\bK_t,\cF_t)\subset\Lnot(\bK_t^0,\cF_t)$   
    for all $t=0,\dots,T$;
  \item[\NRA] (no risk arbitrage) if $\Xi_t^0\cap 
   \Lnot(\R_-^d,\cF_t)=\{0\}$ for all $t=0,\dots,T$;
  \item[\NARA] (no asymptotic risk arbitrage) if $(\cl_0 \Xi_t^0)\cap 
    \Lnot(\R_-^d,\cF_t)=\{0\}$ for all $t=0,\dots,T$;
  \item[\NRAT] (no risk arbitrage opportunity of the second kind) if,
    for any $t=0,\dots,T$ and $\eta_t\in\Lnot(\R^d,\cF_t)$ such that
    $(\eta_t+\Att_{t,T})\cap\Lnot(\bK_T,\cF_T)\neq\emptyset$, we have
    $\eta_t\in\Lnot(\bK_t,\cF_t)+\Acc_{t-1,t}$;
  \item[\SNAR] (strong no risk arbitrage) if $\sum_{t=0}^T
    (k_t+\eta_t)=0$ for $k_t\in\Lnot(\bK_t,\cF_t)$ and
    $\eta_t\in\Acc_{t-1,t}$ for all $t$ implies
    that 
    $k_t\in\Lnot(\bK_t^0,\cF_t)$ and $\eta_t=0$ a.s. for all $t$.
  \end{description}
\end{definition}

Let us comment on the \SNR\ condition.  If $p_t\in
\hat\Xi_t^0\cap\Lnot(-\bK_t,\cF_t)$, then, starting from the zero
initial endowment at time $t$ expressed as $0=p_t-p_t$, we obtain the
zero claim at time $T$ from the price $p_t$ and allowing an immediate
possible profit at time $t$, since the liquidation value of $-p_t\in
K_t$ is non-negative. A similar interpretation applies for the \NRA\
condition and its asymptotic version \NARA. The \NRAT\ condition may
be compared to the \NAT\ condition of \cite{ras08a}, while \SNAR\ is a
version of \cite[Condition (iii), Section 3.2.2]{kab:saf09}. Note that
\NRA\ condition is equivalent to $\Att_{t,T}\cap
\Lnot(\R^d_+,\cF_t)=\{0\}$, while the usual (NA) condition
$\Att_{t,T}\cap \Lnot(\R^d_+,\cF_T)=\{0\}$ is stronger, see
\cite[Section 3.2.1]{kab:saf09}. 

Example~\ref{ex:risk-arbitrage} shows that it may be possible to
release infinite capital from zero position without compromising the
acceptability criterion, in particular, it violates the \SNR\
condition.  
By Lemma~\ref{lemma:st-price}, \SNR\ can be written as
$\bX_t^0\cap(-\bK_t)\subset \bK_t^0$ a.s., and \NARA\ 
as $(\cl\bX_t^0)\cap\R_-^d=\{0\}$ a.s. 
It is obvious that \NARA\ is stronger than \NRA.  By \eqref{eq:1},
\NRA\ condition is equivalent to $\Att_{t,T} \cap
\Lnot(\R_+^d,\cF_t)=\{0\}$ for all $t$. 
 If $\bK_t=\R_+^d$
a.s. for all $t$, then $\Xi_t^0=\Lnot(\R_+^d,\cF_t)$ and 
all no arbitrage conditions are
satisfied, see Example~\ref{ex:univariate}.

\begin{lemma}
  \label{lemma:snr-xi}
  \SNR\ implies that 
  \begin{displaymath}
    \rhos_{t,t+1}(\Xi_{t+1}^0)\cap \Lnot(-\bK_t,\cF_t)
    \subset\Lnot(\bK_t^0,\cF_t),\quad t=0,\dots,T-1.
  \end{displaymath}
  The reverse implication holds if the
  solvency sets are strictly proper.
\end{lemma}
\begin{proof}  
  Denote $M=\rhos_{t,t+1}(\Xi_{t+1}^0)$, $A=\Lnot(\bK_t,\cF_t)$, and
  $B=\Lnot(\bK_t^0,\cF_t)$. It is easily seen that $M\cap
  (-A)\subset B$ if $(M+A)\cap(-A)\subset B$ and only if in case
  $A\cap(-A)=\{0\}$. For the reverse implication, if $x\in
  (M+A)\cap(-A)$, then $x=m+a_1=-a_2$, where $m\in M$ and $a_1,a_2\in
  A$. Therefore, $m/2\in M\cap (-A)\subset B$. Then $x/2\in
  A \cap (-A)$, so that $x\in B=A \cap (-A)=\{0\}$ if $\bK$ is
  strictly proper.
\end{proof}

\begin{lemma}
  \label{lemma:sna-nas-risk}
  Assume that the acceptance sets are strictly proper, that is,
  $\Acc_{t,s}\cap(-\Acc_{t,s})$ consists of all random vectors that
  equal $0$ almost surely.
  \begin{enumerate}[(i)]
  \item If $\bK_t^0=\tilde{\bK}_t$ for all $t$, \SNAR\ implies
  \begin{align}
    \label{eq:9}
    \Att_{0,t}\cap(\Lnot(\bK_t,\cF_t)+\Acc_{t-1,t})
    &\subset\Lnot(\bK_t^0,\cF_t),
    \quad t=0,\dots,T,\\
    \label{eq:10}
    \Att_{t,T}\cap(\Lnot(\bK_t,\cF_t)+\Acc_{t-1,t})
    &\subset\Lnot(\bK_t^0,\cF_t),
    \quad t=0,\dots,T.
    \end{align}
  \item If the solvency sets are strictly proper and 
    \begin{equation}
      \label{eq:8}
      \Lnot(-\bK_t,\cF_t)\cap\Acc_{t-1,t}=\{0\},\quad t=0,\dots,T,
    \end{equation}
    then either of the conditions \eqref{eq:9}, \eqref{eq:10} implies
    \SNAR.
  \end{enumerate}
  \end{lemma}
\begin{proof}
  (i) Motivated by \cite[Lemma~3.2.7]{kab:saf09}, assume that
  \begin{displaymath}
    -k_0-\cdots-k_t-\eta_{0}-\cdots-\eta_{t}\eqdef g_t+\zeta_t\in  \Att_{0,t}\cap
    (\Lnot(\bK_t,\cF_t)+\Acc_{t-1,t}) 
  \end{displaymath}
  with $k_s\in \Lnot(\bK_s,\cF_s)$, $\eta_s\in\Acc_{s-1,s}$,
  $s=0,\dots,t$, $g_t\in \Lnot(\bK_t,\cF_t)$, and $\zeta_t\in
  \Acc_{t-1,t}$. Since $(\eta_t+\zeta_t)/2\in\Acc_{t-1,t}$ by
  convexity and 
  \begin{displaymath}
    -k_0/2-\cdots-k_{t-1}/2-(k_t+g_t)/2-\eta_{0}/2-\cdots-(\eta_t+\zeta_t)/2=0,
  \end{displaymath}
  we have $(k_t+g_t)/2\in \bK_t^0$ and $(\eta_t+\zeta_t)/2=0$ by
  \SNAR. The strict properness of the acceptance sets yields that
  $\eta_t=\zeta_t=0$. Furthermore, $\frac{1}{2}g_t\in
  -\frac{1}{2}k_t+\frac{1}{2}\bK^0_t\subseteq -\bK_t$, so that
  $g_t\in \bK_t^0$, i.e. \eqref{eq:9} holds.

  Property \eqref{eq:10} is similarly derived from \SNAR. 

  (ii) In order to show that \eqref{eq:9} implies \SNAR, proceed by
  induction as in \cite[Lemma~3.2.13]{kab:saf09}. Let
  $-k_0-\cdots-k_T-\eta_0-\cdots-\eta_T=0$. Then
  \begin{displaymath}
    k_T+\eta_T=\sum_{s=0}^{T-1}(-k_s-\eta_s)\in\Att_{0,T-1}\subset\Att_{0,T}.
  \end{displaymath}
  By \eqref{eq:9}, $k_T+\eta_T\in\Lnot(\bK_T^0,\cF_T)$. Since
  $k_T+\eta_T$ is $\cF_{T-1}$-measurable and the solvency sets are
  strictly proper, $k_T+\eta_T\in \Lnot(\bK_{T-1}^0,\cF_{T-1})$. Therefore,
  $k_T+\eta_T$ can be merged with $k_{T-1}$, and then the induction
  proceeds with $T-1$ instead of $T$. 

  To show that \eqref{eq:10} implies \SNAR, proceed by induction
  starting from time zero. Since
  \begin{displaymath}
    k_0+\eta_0=\sum_{s=1}^T (-k_s-\eta_s)\in\Att_{1,T}\subset\Att_{0,T},
  \end{displaymath}
  \eqref{eq:10} yields that $k_0+\eta_0=0$, and \eqref{eq:8} implies
  $k_0=\eta_0=0$.
\end{proof}

Condition \eqref{eq:8} can be viewed as a consistency 
between the acceptance sets and solvency sets, namely, that $-\bK_t$
does not contain any acceptable non-trivial selection.

The first part of the following result shows that \NRA\ is similar to
the weak no arbitrage property NA$^\mathrm{w}$ of Kabanov's model, see
\cite[Sec.~3.2.1]{kab:saf09}. Denote by $\Int A$ the interior and by
$\partial A$ the boundary of $A\subset\R^d$.

\begin{proposition}
  \label{prop-NRA}
  Suppose that $\R^d_+\setminus \{0\}\subseteq\Int\bK_t$ a.s. for all
  $t$. Then \NRA\ is equivalent to each of the following two conditions.
  \begin{enumerate}[(i)]
  \item $\rhonot_{t,t+1}(\Xi_{t+1}^0)\cap
    \Lnot(-\bK_t,\cF_t) \subseteq \Lnot(-\partial \bK_t,\cF_t)$ for
    all $t$.
  \item $\hat \Xi_t^0\cap \Lnot(\R_-^d,\cF_t)=\{0\}$ for all $t$.
  \end{enumerate}
\end{proposition}
\begin{proof}
  (i) Consider $x_t=\gamma_t+k_t$ for $\gamma_t\in
  M\eqdef\rhonot_{t,t+1}(\Xi_{t+1}^0)$ and $k_t\in\Lnot(\bK_t,\cF_t)$.
  Assume that $x_t$ is non-trivial and $x_t\in \Lnot(\R_-^d,\cF_t)$. 
  Hence, $\gamma_t/2=x_t/2-k_t/2\in \Lnot(-\bK_t,\cF_t)$ and
  $\gamma_t/2\in -\Int \bK_t$ on $\{x_t\neq 0\}$, since $\Int\bK_t$
  contains $\R_+^d\setminus \{0\}$, contrary to the assumption.
  
  Consider any $x_t\in M\cap \Lnot(-\bK_t,\cF_t)$ such that $x_t=-k_t$
  for $k_t\in \Lnot(\bK_t,\cF_t)$ such that
  $\P(k_t\in\Int\bK_t)>0$. By a measurable selection argument,
  $k_t+\gamma_t\in \Lnot(\bK_t,\cF_t)$ for some
  $\gamma_t\in\Lnot(\R_-^d,\cF_t)\setminus \{0\}$. Thus,
  \begin{displaymath}
    x_t+k_t+\gamma_t=\gamma_t\in
    (M+\Lnot(\bK_t,\cF_t))\cap \Lnot(\R_-^d,\cF_t),
  \end{displaymath}
  contrary to \NRA.

  (ii) It suffices to show that \NRA\ implies (ii).  Assume that
  $k_t\in \Lnot(\bK_t,\cF_t)$ and
  $\gamma_t\in\rhos_{t,t+1}(\Xi_{t+1}^0)$ are such that
  $k_t+\gamma_t\in\R_-^d$ a.s. and $k_t+\gamma_t\neq 0$ with a
  positive probability. Since $k_t/2+\R_+^d\subset(\{k_t/2\}\cup
  \Int\bK_t)$ a.s., the set $(\Int\bK_t+\rhos_{t,t+1}(\Xi_{t+1}^0))$
  with a positive
  probability has a non-trivial intersection with $\R_-^d$. By \cite[Prop.~2.10]{lep:mol19}
  with $\Gamma=\Int\bK_t$ and $\Xi=\rhonot_{t,t+1}(\Xi_{t+1}^0)$, the
  set $(\Int\bK_t+\rhonot_{t,t+1}(\Xi_{t+1}^0))$ has a non-trivial
  intersection with $\R_-^d$ with a positive probability, contrary to
  \NRA.
\end{proof}

\begin{theorem}
  \label{thr:snr-bnr-nra-i} 
  If the solvency sets are proper, then \SNR\ implies \NARA\ and the
  closedness of $\hat\Xi_t^0$ in probability for all $t=0,\dots,T$.
\end{theorem}
\begin{proof}
  Denote $M\eqdef\rhos_{t,t+1}(\Xi_{t+1}^0)$.  Recall that
  $\cl_0\hat\Xi_t^0=\cl_0\Xi_t^0$.  Assume that $k_t^n+\gamma_t^n\to
  \zeta_t\in\Lnot(\R_-^d,\cF_t)$ a.s.  for $k_t^n\in
  \Lnot(\bK_t,\cF_t)$ and $\gamma_t^n\in M$ such that
  $k_t^n+\gamma_t^n\in\hat\Xi_t^0$, $n\geq1$. Since $M$ is
  $\Lnot$-closed and convex, we may assume by
  \cite[Lemma~2.1.2]{kab:saf09} that $k_t^n\to k_t\in
  \Lnot(\bK_t,\cF_t)$ on the set $A=\{\liminf_n
  \|k_t^n\|<\infty\}$. Hence, $\gamma_t^n\to\gamma_t\in M$, so that
  \begin{displaymath}
    \gamma_t=\zeta_t-k_t\in M\cap \Lnot(-\bK_t,\cF_t)
    \subseteq \Lnot(\bK_t^0,\cF_t).
  \end{displaymath}
  Thus, $\gamma_t\in \bK_t^0$ and $\zeta_t/2=\gamma_t/2+k_t/2\in
  \bK_t$. Hence, $\zeta_t/2\in \R^d_-\cap \bK_t=\{0\}$ and
  $\zeta_t=0$ on $A$.
  
  If $\P(\Omega\setminus A)>0$, assume that
  $k_t^n=\gamma_t^n=\zeta_t=0$ on $A$ by $\cF_t$-decomposability, and
  use the standard normalisation procedure, i.e. divide
  $k_t^n,\gamma_t^n,\zeta_t$ by $(1+\|k_t^n\|)$. Arguing as
  previously, we obtain $k_t\in \Lnot(\bK_t,\cF_t)$ such that
  $\|k_t\|=1$ on $\Omega\setminus A$. Since $0\in M$, we have
  $\gamma_t\in M$ by conditional convexity. Moreover, $k_t+\gamma_t=0$
  since $\zeta_t/(1+\|k_t^n\|)\to 0$. Hence, $\gamma_t\ne 0$ belongs to
  $M\cap \Lnot(-\bK_t,\cF_t)=\{0\}$, which is a contradiction in view
  of Lemma~\ref{lemma:snr-xi}. 

  This argument also yields the closedness of
  $\hat\Xi_{t}^{0}=\Lnot(\bK_{t},\cF_{t})+M$. 
\end{proof}

Let $\overline{\Att}_{t,s}^p$ denote the closure of
$\Att_{t,s}^p\eqdef\Att_{t,s}\cap \Lnot[p](\R^d,\cF_T)$ in $\Lnot[p]$
for $p<\infty$ and in the $\cF_t$-bounded convergence in probability
for $p=\infty$. 

The following theorem states that the \NARA\, and \SNR\, conditions
are weak no arbitrage conditions of the no free lunch type. Recall
that the usual NFL condition is
$\overline{\Att}_{t,T}^p\cap\Lnot(\R^d_+,\cF_T)=\{0\}$.

\begin{theorem} 
  \label{EquivPropNARA} 
  Assume that the acceptance sets are continuous from below and $p\in
  [1,\infty]$.
  \begin{enumerate}[(i)]
  \item \NARA\ is equivalent to
    \begin{equation}
      \label{eq:16}
      \overline{\Att}_{t,T}^p\cap\Lnot(\R^d_+,\cF_t)=\{0\},\quad
      t=0,\dots,T-1.
    \end{equation}
  \item If the solvency sets are proper, then \SNR\ is
    equivalent to
    \begin{equation}
      \label{eq:11}
      \overline{\Att}_{t,T}^p\cap\Lnot(\bK_t,\cF_t)=\{0\},
      \quad t=0,\dots,T-1.
    \end{equation}
  \end{enumerate}
  Moreover, properties (\ref{eq:16}) and (\ref{eq:11}) are
  equivalent to the same ones with $p=1$ and also to those obtained by
  taking the closure of $\Att_{t,T}$ in
  $\Lts[T]=\Lnot[p_{\cF_t}](\R^d,\cF_T)$.
\end{theorem}
\begin{proof}
  (i) Assume that (\ref{eq:16}) holds for the closure with respect to
  the conditional norm on $\Lts[T]$.  Thus, \eqref{eq:16} also holds
  if the closure is taken with respect to the norm on
  $\Lnot[p](\R^d,\cF_T)$. Therefore, given \eqref{eq:1}, it suffices
  to show that \NARA\ follows from
  \begin{equation}
    \label{eq:17}
    \cl_p(\Xi_t^0\cap\Lnot[p](\R^d,\cF_t))
    \cap \Lnot(\R_-^d,\cF_t)=\{0\},\quad t=0,\dots,T-1.
  \end{equation}
  Assume \eqref{eq:17} and consider $x_t\in (\cl_0 \Xi_t^0)\cap
  \Lnot(\R_-^d,\cF_t)$. Then $x_t^n\to x_t$ a.s. for $x_t^n\in
  \Xi_t^0$, $n\geq1$. Hence, 
  \begin{displaymath}
    x_t^n\one_{\|x_t^n\|\le m+1}\one_{\|x_t\|\le m} \to
    x_t\one_{\|x_t\|\le m} \quad\text{a.s. as }\; n\to\infty
  \end{displaymath}
  for all $m\ge 1$, where $x_t^n\one_{\|x_t^n\|\le
    m+1}\one_{\|x_t\|\le m}\in\Xi_t^0$ by decomposability and since
  $0\in\Xi_t^0$. The dominated convergence theorem yields that
  $x_t\one_{\|x_t\|\le m}$ belongs to $\cl_p(\Xi_t^0\cap\Lnot[p](\R^d,\cF_t))\cap
  \Lnot(\R_-^d,\cF_t)=\{0\}$, where the closure may be taken with
  respect to the conditional norm.
  Letting $m\to \infty$ yields $x_t=0$, i.e. \NARA\ holds.

  Assume \NARA. Consider a sequence $(V^n_{t,T})_{n\ge
    1}$ from $\Att_{t,T}^p$ which converges in $\Lts[T]$ to $z_t^+\in
  \Lnot[p](\R^d_+,\cF_t)$. Then, $\tilde V^n_{t,T}=V^n_{t,T}\wedge
  z_t^+\to z_t^+$ in $\LuT$, where the minimum is taken
  coordinatewisely and $u$ is any time moment between $t$ and $T$,
  and $\tilde V^n_{t,T}\in \Att_{t,T}$, so that we may assume without
  loss of generality that $V^n_{t,T}\le z_t^+$.
  Passing to subsequences, assume that $V^n_{t,T}\to
  z_t^+$ in $\LuT$ and almost surely for each given $u\ge t$.
  
  Define $\xi_T^n=V^n_{t,T}-z_t^+\le 0$. Then 
  $\tn\xi_T^n\tn_{p,\cF_{T-1}}\to 0$ in probability.
  By the continuity from below, there exists a sequence
  $\gamma_{T-1}^n\in \Lnot(\R^d_+,\cF_{T-1})$, $n\geq1$, such that
  \begin{displaymath}
    \eta^n_{T}\eqdef \xi^n_T+\gamma_{T-1}^n\in \Acc_{T-1,T}
  \end{displaymath}
  and $0\le \gamma_{T-1}^n\le x_T\tn\xi_T^n\tn_{p,\cF_{T-1}}$ for all
  $n$ and some $x_T\in\R_+^d$. Hence, $\gamma_{T-1}^n\to 0$ in
  $\LuT[T-2]$ if $T-2\ge t$. Since $-\gamma_{T-1}^n\to 0$
  in $\LuT[T-2]$, the continuity from below property yields
  the existence of a sequence $\gamma_{T-2}^n\in
  \Lnot(\R^d_+,\cF_{T-2})$ such that
  \begin{displaymath}
    \eta^n_{T-1}\eqdef-\gamma_{T-1}^n+\gamma_{T-2}^n\in \Acc_{T-2,T-1}
  \end{displaymath}
  for all $n$, and, for some constant $x_{T-1}\in\R_+^d$, we have
  \begin{displaymath}
    0\le\gamma_{T-2}^n\le x_{T-1}\tn\gamma_{T-1}^n\tn_{p,\cF_{T-1}}\le x_Tx_{T-1}
    \tn\xi_T^n\tn_{p,\cF_{T-2}},
  \end{displaymath}
  so that $\gamma_{T-2}^n\to0$ in $\LuT[T-3]$ if $T-3\ge
  t$. Iterate the construction to find
  $\gamma_{T-3}^n,\dots,\gamma_{t}^n$ such that $\gamma_{t}^n\to 0$
  a.s. Then $\eta^n_{u+1}\eqdef-\gamma_{u+1}^n+\gamma_{u}^n\in
  \Acc_{u,u+1}$ if $t\le u\le T-2$. Hence,
  \begin{displaymath}
    \xi_T^n+\gamma_{t}^n=\sum_{u=t}^{T-1}\eta^n_{u+1}\in \Acc_{t,t+1}+\cdots+\Acc_{T-1,T}.
  \end{displaymath}
  By convexity,  
  \begin{displaymath}
    \frac{1}{2}(-z_t^+
    +\gamma_{t}^n)=-\frac{1}{2} V^n_{t,T}+\frac{1}{2}
    (\xi_T^n+\gamma_{t}^n) \in \Xi_t^0.
  \end{displaymath}
  Letting $n\to\infty$ yields that $-\frac{1}{2}z_t^+\in
  (\cl_0\Xi_t^0)\cap\Lnot(-\R^d_+,\cF_t)$, so that $z_t^+=0$ by
  \NARA. Thus, \eqref{eq:16} holds with respect to the conditional
  norm and also with respect to the $\Lnot[p]$-norm.

  (ii) Recall that $\hat\Xi_t^0=\cl_0(\hat\Xi_t^0)=\cl_0(\Xi_t^0)$
  by Theorem~\ref{thr:snr-bnr-nra-i}.  Following the arguments from
  (i), we obtain that \SNR\ is equivalent to
  \begin{equation}
    \label{eq:15}
    \cl_p(\Xi_t^0\cap\Lnot[p](\R^d,\cF_t))\cap
    \Lnot[p](-\bK_t,\cF_t)=\{0\}, \quad t=0,\dots,T-1.
  \end{equation}
  In view of \eqref{eq:1}, $\cl_p(\Xi_t^0\cap\Lnot[p])\subseteq
  -\overline{\Att}_{t,T}^p$.  Therefore, \eqref{eq:11} implies
  \eqref{eq:15} and \SNR\ holds.
  
  Now assume \SNR. Consider a sequence $(V^n_{t,T})_{n\ge
    1}$ from $\Att_{t,T}^p$ which converges in $\Lnot[p]$ to $k_t\in
  \Lnot(\bK_t,\cF_t)$. 
  Then follow the proof of (i) with $k_t$ instead of $z_t^+$.
  
  Consider a sequence $(V^n_{t,T})_{n\ge 1}$ from
  $\Att_{t,T}^1$ which converges in $\Lnot[1]$ to $k_t\in
  \Lnot(\bK_t,\cF_t)$. We may assume that $(V^n_{t,T})_{n\ge 1}$
  converges a.s. Then, for every $M>0$, $(V^n_{t,T}\one_{\|k_t\|\le
    M})_{n\ge 1}$ is a sequence from $\Att_{t,T}^1$ which converges in
  $\Lnot[1]$ to $k_t\one_{\|k_t\|\le M}\in \Lnot[p](\bK_t,\cF_t)$ so
  that we may assume without loss of generality that $\|k_t\|$ is
  bounded by $M$. Passing to a subsequence, assume that
  $\E(\|V^n_{t,T}-k_t\|\,|\cF_t)\to 0$ a.s.  Therefore,
  $V^n_{t,T}\one_{\E(\|V^n_{t,T}\|\,|\cF_t)\le M+1}\to k_t$ 
  almost surely and in $\cL^1_{t,T}$. Thus, (\ref{eq:11}) holds
  with $p=1$.
\end{proof}

Now consider more general claims $\xi$. Recall that, if $\bK_T$ is a cone and
$\xi\in\bK_T$ a.s., then $\Xi_t^\xi\subset\Xi_t^0$ for all $t$.

\begin{theorem}
  \label{thr:snr-bnr-nra-ii} 
  If the solvency sets are proper and the acceptance sets are
  continuous from below, then \SNR\ yields that $\hat\Xi_t^{\xi}$ is
  closed in probability for all $t$ and any $\xi\in\Lnot[p](\R^d,\cF_T)$,
  so that $\hat\Xi_t^\xi=\Lnot(\bX_t^\xi,\cF_t)$ for a random closed
  set $\bX_t^\xi$, $t=0,\dots,T$.
\end{theorem}
\begin{proof}
  Assume that $k_t^n+\gamma_t^n\to\zeta_t\in\Lnot(\R^d,\cF_t)$ a.s. for
  $k_t^n\in \Lnot(\bK_t,\cF_t)$ and $\gamma_t^n\in
  M\eqdef\rhos_{t,t+1}(\Xi_{t+1}^{\xi})$ such that
  $k_t^n+\gamma_t^n\in\hat\Xi_t^0$. Since $M$ is $\Lnot$-closed and
  convex, we may assume by \cite[Lemma~2.1.2]{kab:saf09} that
  $k_t^n\to k_t\in\Lnot(\bK_t,\cF_t)$ on the set $A=\{\liminf_n
  \|k_t^n\|<\infty\}$. Hence, $\gamma_t^n\to\gamma_t\in M$, so that
  $\zeta_t=k_t+\gamma_t\in\hat\Xi_t^{\xi}$.
  
  If $\P(\Omega\setminus A)>0$, assume that
  $k_t^n=\gamma_t^n=\zeta_t=0$ on $A$ by $\cF_t$-decomposability, and
  use the normalisation procedure, i.e. obtain $\tilde k_t^n$ and
  $\tilde\gamma_t^n$ by scaling $k_t^n$ and $\gamma_t^n$ with
  $c_t^n=(1+\|k_t^n\|)^{-1}$. We may assume that
  $\|\tilde\gamma_t^n\|\le 2$, since $c_t^n\zeta_t\to 0$, so that
  $\tilde k_t^n+\tilde\gamma_t^n\to 0$. Arguing as previously, $\tilde
  k_t^n\to\tilde k_t\in\Lnot(\bK_t,\cF_t)$ in $\Lnot[p]$, and
  $\|\tilde k_t\|=1$ on $\Omega\setminus A$. Therefore,
  $\tilde\gamma_t^n\to\tilde\gamma_t=-\tilde k_t$ in $\Lnot[p]$, so
  that $\tilde k_t+\tilde\gamma_t=0$. Notice that
  \begin{displaymath}
    M=\cl_0\left((-\Att_{t+1,T}+\xi)\cap \Lnot(\R^d,\cF_t)\right).
  \end{displaymath}
  By convexity,
  \begin{displaymath}
    \tilde\gamma_t^n\in\cl_0(-\Att_{t+1,T}+c_t^n\xi)\cap\Lnot(\R^d,\cF_t).
  \end{displaymath}
  Since $\|\tilde\gamma_t^n\|\le 2$, assume without loss of
  generality that $\tilde\gamma_t^n\in\cl_p(-\Att_{t+1,T}+c_t^n\xi)$.
  To see this, it suffices to approximate the sequence $\tilde
  \gamma_t^n$ by $\bar\gamma_t^{mn}\in(-\Att_{t+1,T}+c_t^n\xi)$,
  $m\ge1$, and multiply the latter by $\one_{\|\bar\gamma_t^{mn}\|\le
    3}$. Letting $n\to\infty$, \eqref{eq:11} yields that
  \begin{displaymath}
    -\tilde\gamma_t\in\overline{\Att}_{t,T}^p\cap \Lnot(\bK_t,\cF_t)=\{0\}.
  \end{displaymath}
  Thus, $\gamma_t=0$, so that $\P(\Omega\setminus A)=0$ and the
  conclusion follows.
\end{proof}

\begin{lemma}
 \label{lemma:NRAT}
 \NRAT\ is equivalent to
 \begin{equation}
   \label{eq:13}
   \Xi_t^\xi\subset\Lnot(\bK_t,\cF_t)+\Acc_{t-1,t},\quad t=0,\dots,T,
 \end{equation}
 for any $\xi\in\Lnot(\bK_T,\cF_T)$.
 If \eqref{eq:8} holds, then \NRAT\ implies \NARA.
\end{lemma}
\begin{proof}
 Note that $(\eta_t+\Att_{t,T})$ intersects $\Lnot(\bK_T,\cF_T)$ if
 and only if 
 \begin{displaymath}
   \eta_t\in(-\Att_{t,T}+\Lnot(\bK_T,\cF_T))\cap\Lnot(\R^d,\cF_t),
 \end{displaymath}
 equivalently, $\eta_t\in\Xi_t^\xi$ for some $\xi=k_T\in\Lnot(\bK_T,\cF_T)$.
\end{proof}

Denote by $\bK_t^*=\{x:\; \langle x,u\rangle\geq 0, u\in\bK_t\}$ the
positive dual set to $\bK_t$ and assume that $\bK_t^*\setminus\{0\}$
is a subset of the interior of $\R_+^d$ for all $t$.

\begin{definition}
  \label{def:wcps}
  For $t\leq T$, an adapted process $Z\eqdef(Z_s)_{s=t,\dots,T}$ is a
  \emph{$q$-integrable $t$-weakly consistent price system} if it is a
  $\Q$-martingale for $\Q\sim \P$ such that $Z_s$ is a $q$-integrable
  under $\Q$, $\cF_s$-measurable selection of $\bK_s^*$ for every
  $s\ge t$ and $Z_t\neq0$ a.s. We denote
  by $\cM_{t,T}^{q,w}(\Q)$ the set all $q$-integrable $t$-weakly
  consistent price systems under $\Q$, where $q\in[1,\infty]$.
\end{definition}

The following result characterises the prices under \SNR\ and \NARA\
conditions and so may be viewed as the Fundamental Theorem of Asset
Pricing in our framework. 

\begin{theorem}
  \label{coroEquivPropNARA}
  Assume the coherent conical setting and that the solvency sets are
  continuous from below.  Let $q$ be the conjugate of the number $p$ that stems
  from the definition of the acceptance sets.
  \begin{enumerate}[(i)]
  \item \NARA\ is equivalent to the fact that, for each $t$, there
    exists $Z\in\cM_{t,T}^{q,w}(\P)$ such that 
    \begin{equation}
      \label{eq:18}
      \E\langle Z_u,\eta_u\rangle\ge 0 \quad \text{for all }\; 
      \eta_u\in \Acc_{u-1,u}, \; u=t+1,\dots,T.
    \end{equation}
  \item If $\Int\bK_t^*\ne \emptyset$ a.s. for all $t$, then \SNR\ is
    equivalent to the fact that, for each $t$, there exists
    $Z\in\cM_{t,T}^{q,w}(\P)$ such that \eqref{eq:18} holds and 
    $Z_t\in\Lnot(\Int\bK_t^*,\cF_t)$.
  \end{enumerate}
\end{theorem}
\begin{proof} 
  (i) The existence of $Z\in \cM_{t,T}^{q,w}(\P)$ such that $\E
  \langle Z_t,\eta\rangle\ge 0$ for all $\eta\in \Acc_{t,t+1}+\cdots+
  \Acc_{T-1,T}$ is a direct consequence of the Hahn--Banach separation
  theorem and Theorem~\ref{EquivPropNARA}(i), since we may take $p=1$.

  To prove the converse implication, for each $t$, assume the existence of
  $Z\in\cM_{t,T}^{q,w}(\P)$ and consider $x_T\in\Att_{t,T}^p$. Then
  \begin{displaymath}
    x_T=-k_t-(k_{t+1}+\eta_{t+1})-\cdots-(k_T+\eta_T),
  \end{displaymath}
  where $\eta_s\in\Acc_{s-1,s}$ and $k_s\in \Lnot(\bK_s,\cF_s)$ for
  $s\ge t$. Since $\eta_s=\eta'_s+\eta''_s$ with
  $\eta'_s\in\Acc_{s-1,s}\cap\Ltg{s-1}{s}$ and
  $\eta''_s\in\Lnot(\R^d_+,\cF_s)$, we may merge $\eta''_s$ and $k_s$
  and suppose without loss of generality that $\eta_s=\eta'_s$.

  Using the backward induction on $t\le T$, we show that $\E\langle
  Z_T,x_T\rangle\le 0$. If $x_T=-k_{T-1}-k_T-\eta_T$, this is trivial.
  Since $\eta_{t+1},k_{t}\in \Ltg{t}{t+1}$, there
  exists a partition $(B_t^i)_{i\ge 1}$ from $\cF_t$ such that
  $\eta_{t+1}\one_{B_t^i},k_t\one_{B_t^i}\in\Lnot[p](\R^d,\cF_t)$ for all
  $i\ge 1$. Then,
  \begin{displaymath}
    x_{t+1}^i\eqdef(-k_{t+1}-\cdots-k_T-\eta_{t+2}-\cdots-\eta_{T})
    \one_{B_t^i}\in\Att_{t+1,T}^p,\quad i\geq1.
  \end{displaymath}
  Moreover,
  \begin{align*}
    \E \langle Z_T,x_T\rangle
    &=\sum_{i=1}^{\infty}\E \langle Z_T,x_{t+1}^i\rangle
    +\sum_{i=1}^{\infty}\E \langle Z_T,-k_t\one_{B_t^i}\rangle
    +\E \langle Z_t,-\eta_{t+1}\rangle\one_{B_t^i},\\
    &= \sum_{i=1}^{\infty}\E \langle Z_T,x_{t+1}^i\rangle
    +\sum_{i=1}^{\infty}\E \langle Z_t,-k_t\one_{B_t^i}\rangle
    +\E \langle Z_{t+1},-\eta_{t+1}\one_{B_t^i}\rangle\\
    &\le  \sum_{i=1}^{\infty}\E \langle Z_T,x_{t+1}^i\rangle.
  \end{align*}
  The induction hypothesis yields that $\E\langle
  Z_T,x_{t+1}^i\rangle\le 0$. Hence, $\E\langle
  Z_T,x_T\rangle\le0$. Therefore, $\E\langle Z_T,x_T\rangle\le 0$ for
  all $x_T\in\overline{\Att}_{t,T}^p$. In particular, if
  $x_T=x_t\in\Lnot(\R^d_+,\cF_t)$, then $\E\langle Z_t,x_t\rangle\le0$
  and finally $\E\langle Z_t,x_t\rangle= 0$. Since $Z_t\in\Int\R^d_+$,
  we have $x_T=0$, i.e. \NARA\ holds by
  Theorem~\ref{EquivPropNARA}(i).

  (ii) Replicate the proof of (i) using the Hahn-Banach theorem and
  following the arguments of \cite[Th.~4.1]{lep10} in order to
  construct $Z\in\cM_{t,T}^{q,w}(\P)$ such that
  $Z_t\in\Lnot(\Int\bK_t^*,\cF_t)$.
\end{proof}

\begin{remark}
  Condition \eqref{eq:18} can be equivalently written in terms of the
  conditional expectation as $\E(\langle
  Z_u,\eta_u\rangle|\cF_{u-1})\ge 0$ a.s., which also corresponds to
  the duality pairing in modules, see \cite{fil:kup:vog09}. Indeed,
  suppose that \eqref{eq:18} holds.  Then $\E(\langle Z_u,\eta_u
  \one_{A_{u-1}}\rangle) \ge 0$ for all $A_{u-1}\in
  \cF_{u-1}$. Therefore, $\E(\langle Z_u,\eta_u\rangle | \cF_{u-1})\ge
  0$. The opposite implication is obvious. In other words,
  \eqref{eq:18} means that $Z_u$ belongs to the positive dual of
  $\Acc_{u-1,u}$. 
\end{remark}

If the acceptance sets with $p=\infty$ are generated by convex
families $\cZ_{t,s}$ (see Example~\ref{dyn-ex-2}), then \eqref{eq:18}
means that $Z_t$ belongs to the closure of $\cZ_{t,s}$ with respect to
the $\cF_t$-bounded convergence in probability. If $\bK_t$ are all
half-spaces (in the frictionless setting), then \NARA\ is equivalent
to the existence of a martingale $Z_u=\phi_u S_u$, where $S_u$ is the
price vector, such that \eqref{eq:18} holds.

\section{Good Deals hedging}
\label{sec:no-good-deals-1}

Assume that $\Acc_{t,s}$ consists of random vectors
$(\eta_s,0,\dots,0)\in\Lnot[p](\R^d,\cF_s)$ with all vanishing
components except the first one and such that $\risk_{t,s}(\eta_s)\leq 0$ for a
univariate dynamic risk measure $\risk_{t,s}$. This corresponds to the
case, when the acceptability at each step is assessed by calculating
the risk of a portfolio expressed in the units of the first asset,
most importantly, cash. An arbitrage opportunity in this setting is
called a Good Deal, see \cite{cher07fs}.

For simplicity, consider the one period setting with zero interest
rate and two assets exchangeable without transaction costs, so that
without loss of generality the first asset is assumed to be cash and
the second one is a risky asset priced at $S_t$ for $t=0,1$. Let $\xi$
be the cash value of a terminal claim. If $x_0$ is the initial
endowment (in cash), then the terminal value of a portfolio is
\begin{displaymath}
  V_1=(x_0,0)+(-k_0S_0,k_0)+(-k_1S_1,k_1)-(\eta',\eta''),
\end{displaymath}
where $\eta'$ and $\eta''$ are acceptable with respect to some static
convex risk measure $\risk$, meaning that $\risk(\eta')\leq 0$ and
$\risk(\eta'')\leq0$. Given the choice of the acceptance set
$\Acc_{0,1}$, we have 
$\eta''=0$, 
so that $V_1$ suffices to pay the claim if
\begin{displaymath}
  x_0+k_0(S_1-S_0)-\eta'\geq \xi.
\end{displaymath}
The smallest value of $x_0$ that ensures the existence of an
acceptable $\eta'=-\xi+x_0+k_0(S_1-S_0)$, satisfying the above
inequality, equals the infimum of $\risk(k_0(S_1-S_0)-\xi)$ over all
deterministic $k_0\in\R$. For instance, the zero claim $\xi=0$ can be
hedged with a negative initial capital if $\risk(S_1-S_0)<0$ or
$\risk(S_0-S_1)<0$, and this means the existence of a Good Deal
arbitrage. If this is the case and the risk measure is coherent, then
also
\begin{equation}
  \label{eq:2a}
  \risk(k_0(S_1-S_0)-\xi)\leq \risk(k_0(S_1-S_0))+\risk(-\xi)<0
\end{equation}
for sufficiently large positive $k_0$ if $\risk(S_1-S_0)<0$ (or
negative $k_0$ if $\risk(S_0-S_1)<0$), meaning that any claim with
finite $\risk(-\xi)$ can be also hedged with a negative initial
investment.  In other words, the No Good Deals (NGD) Arbitrage
condition becomes
\begin{displaymath}
  \risk(S_1/S_0)\geq -1 \quad \text{and}\quad \risk(-S_1/S_0)\geq 1. 
\end{displaymath}

Our setting is more general as the Good Deals hedging, since it allows
for more general acceptance sets and eliminates the prescribed choice
of the single asset in order to assess the acceptability. As a result,
the no arbitrage conditions become stronger and the superhedging price
declines.  To illustrate this, consider the above two assets one
period setting with the acceptance set $\Acc_{0,1}$ that consists of
all $(\eta',\eta'')$ such that $\risk(\eta')\leq 0$ and
$\risk(\eta'')\leq 0$. By allowing a non-trivial $\eta''$, it is
possible to decrease the price of a terminal cash claim $\xi$. For
this, note $\xi$ can be paid if
\begin{align*}
  \begin{cases}
    x_0-k_0S_0-k_1S_1-\eta'&\geq \xi\\
    k_0+k_1-\eta'' &\geq 0
  \end{cases}
\end{align*}
for some deterministic $k_0$, $\cF_1$-measurable $k_1$, and acceptable
$\eta',\eta''$. This increases the hedging possibilities and so leads
to a decrease of the superhedging price, but also creates extra
arbitrage opportunities. In particular, considering
$(\eta'.\eta'')\in\Acc_{0,1}$ with $\eta'=0$, the arbitrage becomes
possible if $\risk((k_0(S_1-S_0)+x_0)/S_1)\leq 0$ for some $x_0<0$ and
$K_0\in\R$. By letting $x_0$ increase to zero, we see that the
necessary no arbitrage condition in addition to \eqref{eq:2a} yields
that
\begin{equation}
  \label{eq:2b}
  \risk(S_0/S_1)\geq -1 \quad \text{and}\quad \risk(-S_0/S_1)\geq 1. 
\end{equation}
It corresponds to the fact that in two assets, the position expressed
in one of them may be not acceptable, while the position expressed in
the other one may be acceptable. The necessary and sufficient no
arbitrage condition is stronger and should also include all possible
combinations of the two assets.

Assume that $\xi=(S_1-K)^+$ for some $K>0$ and that the support of
$S_1$ is the whole half-line $(0,\infty)$. If $\risk(X)=\esssup[\cF_0](-X)$,
i.e. when the acceptable positions are non-negative random variables,
then the minimal price
\begin{displaymath}
  x_0=\inf_{k_0\in \R}\risk(k_0(S_1-S_0)-\xi)
\end{displaymath}
equals $S_0$. If $\risk$ is non-trivial, we have $x_0\le
S_0+\risk(S_1-\xi)$. Note that $S_1-\xi=S_1\wedge K$, so that $x_0\le
S_0+\risk(S_1)\wedge K$ and, finally, $x_0\le S_0-\risk(S_1)$, where
$\risk(S_1)<0$ given that $S_1>0$ and $\risk$ is non-trivial. This
simple example illustrates the decrease of the super-hedging price in
presence of a non-trivial risk measure.  

\begin{example}
  Assume that the risk measure $\risk$ is the negative essential
  infimum, that is, consider the setting of conditional cores from
  Section~\ref{sec:conditional-core-as}. Then the NGD arbitrage is not
  possible if $\essinf[\cF_0] S_1\leq S_0\leq \esssup[\cF_0] S_1$. With this choice
  of the risk measure, the NGD condition coincides with the \SNR\
  condition, see Theorem~\ref{NAS-2d}. 
\end{example}

\section{Conditional core as risk measure}
\label{sec:conditional-core-as}

Assume that $p=\infty$ and $\Acc_{t,s}=\Lnot(\R_+^d,\cF_s)$ for all
$0\leq t\leq s\leq T$, so that
$\rhonot_{t,s}(\Xi)=\Xi\cap\Lnot(\R^d,\cF_t)$ for any upper set
$\Xi\subset\Lnot(\R^d,\cF_s)$. If $\bX$ is an upper random closed set,
then
\begin{displaymath}
  \rhonot_{t,s}(\bX)=\rhos_{t,s}(\bX)=\cm(\bX|\cF_t),
\end{displaymath}
where the latter notation designates the largest $\cF_t$-measurable
random closed subset of $\bX$, called the \emph{conditional core} of
$\bX$, see \cite[Def.~4.1]{lep:mol19}.  An
acceptable portfolio process 
is characterised by $V_{t-1}-V_t\in\bK_t$ a.s. for $t=1,\dots,T$.
Then $\Att_{t,s}$ becomes the sum of $\Lnot(-\bK_u,\cF_u)$ for
$u=t,\dots,s$, exactly like in the classical theory of markets with
transaction costs \cite{kab:saf09}.
For claim $\xi$, the set $\Xi_t^\xi$ defined in
Section~\ref{sec:hedg-with-coher} becomes the set of superhedging
prices that was used in \cite{loeh:rud14} to define a risk measure of
$\xi$.

The classical no arbitrage condition \NAS\ (no strict arbitrage
opportunity at any time, see \cite[Sec.~3.1.4]{kab:saf09}) then
becomes \eqref{eq:9}; \SNA\ (strong no arbitrage, see Condition~(iii)
in \cite[Sec.~3.2.2]{kab:saf09}) then becomes the special case of
\SNAR.

\begin{theorem} 
  \label{NAS-weakArbitrage}
  Suppose that the solvency sets $(\bK_t)_{t=0,\dots,T}$ are strictly
  proper. Then \SNR, \NAS\ and \SNA\ are all equivalent and are also
  equivalent to each of the following conditions.
  \begin{enumerate}[(i)]
  \item $\overline{\Att}_{t,T}^p\cap \Lnot(\bK_t,\cF_t)=\{0\}$, for
    all $t\le T-1$.
  \item $\Att_{t,T}$ is closed in $\Lnot$ and $\Att_{t,T}\cap
    \Lnot(\bK_t,\cF_t)=\{0\}$, for all $t\le T-1$.
  \end{enumerate}
\end{theorem}
\begin{proof}
  \SNR\ is equivalent to (i) by Theorem~\ref{EquivPropNARA}(ii).  The
  equivalence of \NAS\ and \SNA\ follows from
  Lemma~\ref{lemma:sna-nas-risk} given that \eqref{eq:8} trivially
  holds.

  The implication (i)$\Rightarrow$(ii) is simple to show by
  induction. First, $\Att_{T,T}$ is closed. Assume that
  $-k_t^n-\cdots-k_T^n\to \xi$ a.s. for $k^n_u\in\Lnot(\bK_u,\cF_u)$,
  $u\ge t$. On the set $\{\liminf_n\|k_t^n\|=\infty\}$, we use the
  normalisation procedure to arrive at a contradiction with
  (i). Otherwise, suppose that $-k_t^n\to -k_t\in -\bK_t$, so that we
  may use the induction hypothesis to conclude.

  In order to derive the closedness of $\Att_{t,r}$ under \SNA, it
  suffices to follow the proof of
  \cite[Lemma~3.2.8]{kab:saf09}. Indeed, since $\bK_t^0$ is a linear
  space, the recession cone 
  \begin{displaymath} 
    \bK_t^\infty\eqdef \bigcap_{\alpha>0}\alpha \bK_t=\{x\in \R^d:\;
    \bK_t+\alpha x\subseteq \bK_t\;\forall \alpha>0\}
  \end{displaymath}
  satisfies $\bK_t^0\subseteq \bK_t^{\infty}$, see \cite{pen:pen10}.
  Therefore, $k_t+\alpha x\in \bK_t$ for all
  $k_t\in \bK_t$, $x\in \bK_t^0$, and all $\alpha\in\R$.  Furthermore,
  \SNA\ trivially implies $\Att_{t,T}\cap
  \Lnot(\bK_t,\cF_t)=\{0\}$ for all $t$.

  In order to show that (ii) implies \NAS, assume
  \begin{displaymath}
    -k_0-\cdots-k_t=\tilde k_t\in \Att_{0,t}\cap
    \Lnot(\bK_t,\cF_t).
  \end{displaymath}
  Then $k_0\in \Att_{0,T}\cap \Lnot(\bK_0,\cF_0)$, i.e. $k_0=0$ by
  (ii). Similarly, $k_1=\cdots=k_{t-1}=0$, so that $-k_t=\tilde
  k_t=0$, since $\bK_t$ is strictly proper. Thus, \NAS\ holds.

  At last, \NAS\ yields \SNA, so that $\Att_{t,T}$ is closed in
  $\Lnot$. Finally, \SNA\ yields \eqref{eq:10} and so
  $\Att_{t,T}\cap\Lnot(\bK_t,\cF_t)=\{0\}$, that is, (i) holds. 
\end{proof}

\NAT\ (no arbitrage opportunity of the second kind) from \cite{ras08a}
and \cite[p.~135]{kab:saf09} has the same formulation as \NRAT. 

\begin{lemma}
  \label{lemma:nat-implies-inclusion} 
  Assume that the solvency sets are cones. 
  \begin{enumerate}[(i)]
  \item \NAT\ is equivalent to $\Xi_t^0=\Lnot(\bK_t,\cF_t)$ for all
    $t\le T$.
  \item \NAT\ is equivalent to
    \begin{equation}
      \label{eq:ras-2}
      \cm(\bK_t|\cF_{t-1})\subseteq \bK_{t-1},\qquad t=1,\dots,T.
    \end{equation}
  \item If the solvency sets are strictly proper, then \NAT\ implies
    \SNR.
  \end{enumerate}
\end{lemma}
\begin{proof}
  (i) By Lemma~\ref{lemma:NRAT} and Lemma~\ref{lemma:st-price}(iv), 
  $\Xi_t^0=\Lnot(\bK_t,\cF_t)$ yields \eqref{eq:13}, and so implies
  \NAT. In the other direction, \NAT\ yields that $\Xi_t^0\subset
  \Lnot(\bK_t,\cF_t)+\Acc_{t-1,t}$, while \eqref{RelPortfolio} and the
  choice of $\Acc_{t-1,t}$ yields that
  $\Xi_t^0\supset\Lnot(\bK_t,\cF_t)$. 
  
  (ii) If $\Xi_t^0\subset\Lnot(\bK_t,\cF_t)$,  then  
  \begin{displaymath}
    \Xi_{t+1}^0\cap \Lnot(\R^d,\cF_t)\subset \Lnot(\bK_t,\cF_t)
  \end{displaymath}
  for all $t$ by \eqref{RelPortfolio}. Since
  $\Lnot(\bK_{t+1},\cF_{t+1})\subset\Xi_{t+1}^0$, we 
  obtain \eqref{eq:ras-2}.  

  If \eqref{eq:ras-2} holds, then
  \begin{displaymath}
    \Xi_{T-1}^0=\Lnot(\bK_{T-1},\cF_{T-1})
    +\cm(\bK_T|\cF_{T-1})\subset\Lnot(\bK_{T-1},\cF_{T-1}).
  \end{displaymath}
  Assume that $\Xi_s^0\subset\Lnot(\bK_s,\cF_s)$ for
  $s=t+1,\dots,T$. Then 
  \begin{align*}
    \Xi_t^0&=\Lnot(\bK_{t},\cF_{t})+(\Xi_{t+1}^0\cap\Lnot(\R^d,\cF_t))\\
    &\subset \Lnot(\bK_{t},\cF_{t})+\cm(K_{t+1}|\cF_t)\subset
    \Lnot(\bK_{t},\cF_{t}). 
  \end{align*}
  The proof is finished by the induction argument. 

  (iii) Since $\Xi_t^0$ is closed in probability under \SNR, we have 
  $\Xi_t^0=\hat\Xi_t^0$, and \SNR\ yields that
  $\Lnot(\bK_t,\cF_t)\cap\Lnot(-\bK_t,\cF_t)=\{0\}$, which is the case
  if the solvency sets are strictly proper. 
\end{proof}

For conical solvency sets satisfying $\cm(\bK_t^0|\cF_{t-1})\subseteq
\bK_{t-1}^0$, $t\le T$, in particular, for strictly proper ones, \NAS\
is equivalent to the existence of a $\Q$-martingale evolving in the
relative interiors of $(\bK_t^*)_{t=0,\dots,T}$ for a probability
measure $\Q$ equivalent to $\P$, see \cite[Th.~3.2.2]{kab:saf09}.
Such a martingale is called a \emph{strictly consistent price system}.
If $\Int\bK_t^*\neq\emptyset$ for all $t$, this result follows from
Theorem~\ref{coroEquivPropNARA}(ii).

Note that $\Xi_T^\xi=\Lnot(\bX_T^\xi,\cF_T)$ for $\bX_T=\xi+\bK_T$,
and $\Xi_{T-1}^\xi= \Lnot(\bX_{T-1}^\xi,\cF_{T-1})$ is the family of
selections for a (possibly non-closed) random set
$\bX_{T-1}^\xi=\bK_{T-1}+\cm(\bX_T^\xi|\cF_{T-1})$.  One needs
additional assumptions of the no arbitrage type in order to extend
this interpretation for $\Xi_t^\xi$ with $t\leq T-2$.  Precisely the
sum above should be closed, so that $\cm(\bX_t^\xi|\cF_{t-1})$ exists
for $t\le T-1$, which makes it possible to apply
Lemma~\ref{A+B-measurable}.

\begin{theorem}
  \label{thr:Xt-1} 
  Assume that the solvency sets are strictly proper and \NAS\
  (equivalently, \SNR\ or \SNA) holds. Then
  $\Xi_t^\xi=\Lnot(\bX_t^\xi,\cF_t)$, where $\bX_t^\xi$ is an
  $\cF_t$-measurable random closed convex set, $t=0,\dots,T$, such
  that $\bX_T^\xi=\xi+\bK_T$, and
  \begin{equation}
    \label{eq:xt-1}
    \bX_{t}^\xi\eqdef \bK_{t}+\cm(\bX_{t+1}^\xi|\cF_{t}),
    \qquad t=T-1,\dots,0.
  \end{equation}
\end{theorem}
\begin{proof} 
  It suffices to confirm the statement for $t=T-1$ and then use the
  induction. Indeed, by Theorem~\ref{NAS-weakArbitrage}, \NAS\ is equivalent
  to \SNR\ so that  Theorem~\ref{thr:snr-bnr-nra-ii} applies.  
  Since $\Xi_t^\xi$ is $\cF_{t}$-decomposable,
  Theorem~\ref{ReprDecompL0} yields the existence of an
  $\cF_{T-1}$-measurable closed set $\bX_{T-1}^{\xi}$ such that
  $\Xi_{T-1}^\xi=\Lnot(\bX_{T-1}^{\xi},\cF_{T-1})$. Since $\bX_T^\xi$
  is closed,
  \begin{align*}
    \Xi_{T-1}^\xi&=\Lnot(\bK_{T-1},\cF_{T-1})
    +\Lnot(\cm(\bX_T^\xi|\cF_{T-1}),\cF_{T-1}),\\
    &=\Lnot(\bK_{T-1}+\cm(\bX_T^\xi|\cF_{T-1}),\cF_{T-1})\\
    &=\Lnot(\bX_{T-1}^\xi,\cF_{T-1}),
  \end{align*}
  where $\bX_{T-1}^\xi$ is a random set by
  Lemma~\ref{A+B-measurable}. 
\end{proof}

\begin{proposition}
  \label{PropNA2-equiv} 
  Suppose that the solvency sets are strictly proper. Then \NAS\ holds
  if and only if $\Xi_t^0=\Lnot(\bX_t^0,\cF_t)$ for random closed sets
  $(\bX_t^0)_{t=0,\dots,T}$ such that $\bX_t^0\cap(-\bK_t)=\{0\}$
  a.s. for all $t$. In the conical case, the latter
  condition is equivalent to $\Int(\bX_t^0)^*\ne\emptyset$ for all
  $t$, and, under \NAS,
  \begin{equation}
    \label{eq:6}
    \Att_{t,T}=\sum_{s=t}^T\Lnot(-\bX_s^0,\cF_s), 
    \quad 0\leq t\le T,
  \end{equation}
  where $\bX_t^0$ is a strictly proper random closed convex cone for
  all $t$.
\end{proposition}
\begin{proof} 
  Assume \NAS, so that Theorem \ref{thr:Xt-1} applies. Let
  $-g_t\in\Lnot(\bX_t^0\cap(-\bK_t),\cF_t)$. Then $g_t\in\bK_t$ a.s.,
  and there exist $k_u\in\Lnot(\bK_u,\cF_u)$, $u=t,\dots,T$, and
  $\tilde g_T\in \Lnot(\bK_T,\cF_T)$, such that
  $-g_t-k_t-k_{t+1}-\cdots-k_T=\tilde g_T$. Since \SNA\ holds,
  $g_t+k_t=0$, and $g_t=0$.  The reverse implication is trivial. In
  the conical case, since $K_t+K_t=K_t$ for all $t\le T$, \eqref{eq:6}
  follows from the inclusions 
  \begin{displaymath}
    \bK_t\subseteq \bX_t^0\subseteq
    \bK_t+\cdots+\bK_T,\quad  t\leq T.
  \end{displaymath}
  Note that $\bX_T^0=\bK_T$ is strictly proper by assumption. Since
  \begin{displaymath}
    \bX_{t-1}^0=\bK_{t-1}+\cm(\bX_t^0|\cF_{t-1})
    \subseteq \bK_{t-1}+\bX_t^0,
  \end{displaymath}
  the induction argument yields that $\bX_t^0\subseteq
  \bK_t+\cdots+\bK_T$.  Since \SNA\ holds under \NAS, $\bX_t^0$ is
  strictly proper for all $t$.

  By \cite[Lemma~5.1.2]{kab:saf09}, \NAS\ holds if and only if
  $\Xi_t^0$ is closed and $\Xi_t^0=\Lnot(\bX_t^0,\cF_t)$ with
  \begin{displaymath}
    \Int\bK_t^*\cap
    \Int (\bX_t^0)^*=\Int\bK_t^*\cap\Int\,\cm(\bX_{t+1}^0|\cF_t)^*\ne
    \emptyset,\qquad t\le T.
  \end{displaymath}
  Finally, observe that 
  \begin{displaymath}
    \Int\bK_t^*\cap
    \Int\,\cm(\bX_{t+1}^0|\cF_t)^*=\Int\bK_t^*\cap
    \cm(\bX_{t+1}^0|\cF_t)^*=\Int (\bX_t^0)^*. 
  \end{displaymath}
\end{proof}

Equation \eqref{eq:6} means that, in the superhedging problem, we may
replace solvency sets $\bK_t$ with $\bX_t^0$. The solvency sets
$(\bX_t^0)_{t=0,\dots,T}$ satisfy \NAT\ condition by
Lemma~\ref{lemma:nat-implies-inclusion}, which is generally required
to obtain a dual characterisation of the superhedging prices, see
Condition {\bf B} (equivalent to \NAT) in \cite[Sec.~3.6.3]{kab:saf09}.
Therefore, \NAS\ suffices for \cite[Th.~3.6.3]{kab:saf09} to hold
provided that we consider the consistent price systems associated to
$(\bX_t^0)_{t=0,\dots,T}$.

Now consider \NRA\ and \NARA\ condition for the chosen acceptance
sets.  Assume that the solvency sets are conical and satisfy
$\bK_t^*\setminus\{0\}\subseteq\Int\R^d_+$.
Since
$\hat\Xi_t^{0}=\cl_0(\hat\Xi_t^{0})=\cl_0(\Xi_t^{0})=\Xi_t^{0}$, \NRA\
and \NARA\ are equivalent by Proposition~\ref{prop-NRA}.  Denote by
$\Att_{t,T}^p(\Q)$ and $\overline{\Att}_{t,T}^p(\Q)$ for
$p\in[1,\infty]$ the variants of $\Att_{t,T}^p$ and
$\overline{\Att}_{t,T}^p$ when the reference probability measure is
$\Q$. 

\begin{proposition}
  \label{EqNWA} 
  The following statements are equivalent.
  \begin{enumerate}[(i)]
  \item $\overline{\Att}_{t,T}^p(\Q)\cap \Lnot(\R^d_+,\cF_t)=\{0\}$,\, for
    all $t\le T-1$, $p\in[1,\infty)$ and $\Q\sim \P$.
  \item $\cM_{t,T}^{q,w}(\Q)\neq\emptyset$ for every $t\le T-1$,
    $\Q\sim \P$ and $p\in[1,\infty)$.
  \item \NARA.
  \item $\cM_{t,T}^{\infty,w}(\P)\neq\emptyset$ for every $t\le T-1$.
  \item $\cM_{t,T}^{1,w}(\P)\neq\emptyset$ for every $t\le T-1$.
  \end{enumerate}
\end{proposition}
\begin{proof} 
  By Theorem~\ref{coroEquivPropNARA}, (v) and (iii) are equivalent.
  We deduce the equivalence of (ii) and (iv) by following the proof of
  \cite[Lemma~3.2.4]{kab:saf09}, which makes it possible to construct
  a (weakly)-consistent price system (see Definition~\ref{def:wcps})
  from any consistent price system in $\Lnot[1]$. In particular, (v)
  implies (iv) and, clearly, (iv) implies (v). Then (iii) implies
  (ii), i.e. \NARA\ holds for $Q$ in place of $P$. By
  Theorem~\ref{EquivPropNARA}, (i) holds. At last, (i)
  implies \NARA\ by Theorem \ref{EquivPropNARA}.
\end{proof}

\section{Arbitrage with acceptable expectations}
\label{sec:sandwich-theorems}

Assume that $p=1$, and let $ \Acc_{t,s}\cap \Ltsp[1]$ be the set of all
$\eta_s\in \Ltsp[1]$ such that $\E^g(\eta_s|\cF_t)$ has all
non-negative components. 
In other words, the acceptable positions are those having non-negative
generalised conditional expectation. This is the weakest possible
acceptability criterion, which is always the case (in the static
setting) if the acceptance sets are dilatation monotonic. 

The generalised conditional expectation is well defined for each
$\gamma_s\in \Ltsbarp$ by letting
$\E^g(\gamma_s|\cF_t)=\E^g(\gamma'_s|\cF_t)+\E(\gamma''_s|\cF_t)$, where the
expectation of $\gamma''_s\in\Lnot(\R_+^d,\cF_s)$ may be infinite. If
$\bX_T$ is an $\cF_T$-measurable random upper convex set that admits at
least one selection from $\Ltsp[1]$, then let
\begin{displaymath}
  \rhonot_{t,s}(\bX_T)=\big\{\E^g(\gamma_s|\cF_t):\;
  \gamma_s\in\Lnot[1_{\cF_t}](\bX_T,\cF_s)\big\},
\end{displaymath}
and $\rhos_{t,s}(\bX_T)=\E^g(\bX_T|\cF_t)$ is the generalised conditional
expectation of the random closed set $\bX_T$, see \cite[Def.~6.3]{lep:mol19}. 

Let $\xi\in\Lnot[1](\R^d,\cF_T)$.  By Lemma~\ref{lemma:st-price},
$\hat \Xi_{T-1}^\xi$ is the family of selections of the (possibly,
non-closed) random set
\begin{displaymath}
  \bX_{T-1}^\xi=\bK_{T-1}+\E(\xi|\cF_{T-1})+\E(\bK_T|\cF_{T-1})\,,
\end{displaymath}
which is $\cF_{T-1}$-measurable by Lemma~\ref{A+B-measurable}. Note
that all solvency sets are integrable and so their generalised
conditional expectation coincides with the usual one.  Therefore,
\begin{align*}
  \rhos_{T-2,T-1}(\hat \Xi_{T-1}^\xi)
  &=\E(\bX_{T-1}^\xi|\cF_{T-2})\\
  &=\E(\xi|\cF_{T-2})+\E(\bK_{T-1}+\E(\bK_t|\cF_{T-1})|\cF_{T-2}),\\
  &=\E(\xi|\cF_{T-2})+\E(\bK_T+\bK_{T-1}|\cF_{T-2}).
\end{align*}
Since $k_T\in\Lnot[1_{\cF_{T-1}}](\bK_T,\cF_T)$ and
\begin{multline*}
  \rhonot_{T-2,T-1}(\Xi_{T-1}^\xi)\\
  =\big\{\E^g(k_{T-1}+\E^g(\xi+k_T|\cF_{T-1})|\cF_{T-2}):\; 
  k_{T-1}\in\Lnot(\bK_{T-1},\cF_{T-1})\big\},
\end{multline*}
we deduce that
$\rhos_{T-2,T-1}(\hat\Xi_{T-1}^\xi)\subseteq\rhos_{T-2,T-1}(\Xi_{T-1}^\xi)$. 
Since $\Xi_{T-1}^\xi$ is a subset of $\hat \Xi_{T-1}^\xi$, we have
\begin{displaymath}
  \rhos_{T-2,T-1}(\hat\Xi_{T-1}^\xi)=\rhos_{T-2,T-1}(\Xi_{T-1}^\xi).
\end{displaymath}
Therefore,
\begin{displaymath}
  \bX_{T-2}^\xi=\bK_{T-2}+\E(\xi|\cF_{T-2})
  +\E(\bK_T+\bK_{T-1}|\cF_{T-2}),
\end{displaymath}
Continuing recursively, we obtain $\hat\Xi_t^\xi=\Lnot(\bX_t^\xi,\cF_t)$ with
a not necessarily closed $\cF_{t}$-measurable random set
\begin{equation}
  \label{eq:7}
  \bX_t^\xi=\bK_t+\E(\xi|\cF_t)
  +\E\big(\bK_T+\bK_{T-1}+\cdots+\bK_{t+1}|\cF_t\big).
\end{equation}
Notice that $\bX_t^\xi=\bX_t^0+\E(\xi|\cF_t)$, i.e. $\bX_t^0$
determines all superhedging prices.  Reformulating requirements from
Definition~\ref{def:nra-multiperiod}, we arrive at the following
result.

\begin{proposition}
  \label{SNR-CondExp}
  For the risk arbitrage conditions formulated for the conditional
  expectation as the risk measure, the following hold. 
  \begin{enumerate}[(i)]
  \item  If the solvency sets are strictly proper, \SNR\ is
      equivalent to 
    \begin{displaymath}
      \E(\bK_{t+1}+\cdots+\bK_T|\cF_t)\cap(-\bK_t)=\{0\},
      \;\;\text{a.s.},\quad t=0,\dots,T-1.
    \end{displaymath}
  \item \NARA\ is equivalent to     
    \begin{displaymath}
      \big(\bK_t+\E(\bK_{t+1}+\cdots+\bK_T|\cF_t)\big)\cap
      \R^d_-=\{0\} \;\; \text{a.s.},\quad t=0,\dots,T-1.
    \end{displaymath}
  \end{enumerate}
\end{proposition}

Note that statement (ii) above follows from (i).
Theorem~\ref{coroEquivPropNARA} yields the following result.

\begin{proposition} 
  \label{prop:ce}
  Assume that the solvency sets $(\bK_t)_{t=0,\dots,T}$ are cones. Then
  \NARA\ (resp. \SNR) is equivalent to the existence of a
  deterministic point $z\neq0$ that belongs to all $\bK_t^*$
  (resp. $\Int\bK_t^*$), $t=0,\dots,T$.
\end{proposition}
\begin{proof}
  Any acceptable position from $\Acc_{u-1,u}$ is of the form 
  \begin{displaymath}
    \eta_{u}=
    \big[\gamma_{u}-\E^g(\gamma_{u}|\cF_{u-1})\big]+\E^g(\gamma_{u}|\cF_{u-1})+\zeta^+_{u}
  \end{displaymath}
  with $\E^g(\gamma_{u}|\cF_{u-1})\in \R^d_+$ 
  and
  $\zeta^+_{u}\in\Lnot(\R_+^d,\cF_{u})$. Thus, 
  \begin{align*}
    \E^g\big(\langle Z_u,\eta_u\rangle|\cF_{u-1}\big)
    & \ge \E^g(\langle
    Z_u,\gamma_{u}-\E^g\big(\gamma_{u}|\cF_{u-1})|\cF_{u-1}\big)\\
    & = \E^g\big(\langle
    Z_{u},\gamma_u\rangle|\cF_{u-1}\big) -\langle
    Z_{u-1},\E^g(\gamma_{u}|\cF_{u-1})\rangle.
  \end{align*}
  Hence, $\E^g(\langle Z_u,\eta_u\rangle|\cF_{u-1})\ge 0$ if there
  exists $Z\in\cM_{t,T}^{\infty,w}(\P)$, such that
  \begin{displaymath}
    \E^g\big(\langle Z_u,\gamma_{u}\rangle|\cF_{u-1}\big)
    = \langle Z_{u-1},\E^g(\gamma_{u}|\cF_{u-1})\rangle\quad \text{a.s.}
  \end{displaymath}
  for all $\gamma_{u}\in\Lnot[1_{\cF_{u-1}}](\R^d,\cF_{u})$. The
  equality follows from \eqref{eq:18} by taking
  unconditional expectation (restricting to a partition if necessary)
  and applying the same reason with $-\gamma_u$. 
  Given that $Z_u$ is essentially bounded, it is possible to let
  $\gamma_u$ be equal to one of the component of $Z_u$ multiplied by
  the corresponding basis vector. Thus, the square of every component
  of $(Z_u)_{u=0,\dots,T}$ is a martingale,
  whence $Z_u$ equals to deterministic $z$ for all $u$.  

  The inverse implication follows from
  Theorem~\ref{coroEquivPropNARA}(i) 
  applied to 
  \begin{displaymath}
    \eta_u=\gamma_{u}-\E^g(\gamma_{u}|\cF_{u-1})\in \Acc_{u-1,u}. 
  \end{displaymath}
  The proof for \SNR\ follows from the same argument and
  Theorem~\ref{coroEquivPropNARA}(ii). 
\end{proof}

\begin{remark}
  It is possible to derive the result of Proposition~\ref{prop:ce}
  from Proposition~\ref{SNR-CondExp} by using the fact that the
  expectation of the cone $\bK_t$ is the whole space unless $\bK_t^*$
  contains a deterministic point $z$ distinct from the origin and then
  $\E\bK_t$ is a subset of the half-space with outer normal $(-z)$. In order
  that $\bX_t^0$ does not intersect $\R_-^d$, all cones $\bK_t$ should
  have non-trivial expectation and the sum of these expectations has
  to be non-trivial. This amounts to the existence of a deterministic
  point $z\neq0$ that belongs to $\bK_0^*\cap\cdots\cap \bK_T^*$.
\end{remark}

\section{Application to the  two-dimensional model}
\label{sec:appl-two-dimens}

Consider a financial market model composed of two assets. The first
one has constant value $1$ and the second one is a risky asset
modelled by a bid-ask spread $\bY_t=[S^b_t,S^a_t]$ such that
$0<S^b_t\le S^a_t$ a.s. for all $t\le T$. This is Kabanov's model with
the conical solvency set $\bK_t=C(\bY_t)$, where $C([s',s''])$ is the
positive dual to the smallest cone in $\R^2$ containing the set
$\{1\}\times[s',s'']$.

Consider the acceptance sets from
Section~\ref{sec:conditional-core-as}, so that the conditional core is
the risk measure.
Then $\bX_{T-1}^0$ is the sum of $\bK_{T-1}$ and
$\cm(\bX_T^0,\cF_{T-1})=C(\CM(\bY_T|\cF_{T-1}))$, where
$\CM(\bY_T|\cF_{T-1}))$ is the conditional convex hull of $\bY_T$,
that is, the smallest $\cF_{T-1}$-measurable random closed convex set
that contains $\bY_{T-1}$, see \cite[Def.~5.1]{lep:mol19}. 

Since $\bX_{T-1}^0$ is a random closed set, iterating this argument
yields that $\bX_{t}^0=C(\tilde \bY_{t})$ for $t=0,\dots,T$, where
$\tilde \bY_T=\bY_T$ and
\begin{displaymath}
 \tilde  \bY_{t}=\CM(\tilde \bY_{t+1}|\cF_{t})\cap \bY_{t},\quad
  t=T-1,\dots,0.
\end{displaymath}
Note that we do not make any no arbitrage assumption to obtain
$\bX_t^0$.  Observe that $\tilde\bY_{t}=[\tilde S^b_{t}, \tilde
S^a_{t}]$, where $\tilde S^a_T=S^a_T$, $\tilde S^b_T=S^b_T$, and
\begin{align}
  \label{IterTildeS}
  \tilde S^a_{t}&=S^a_{t}\wedge \esssup[\cF_{t}]\tilde S^a_{t+1},\quad 
  \tilde S^b_{t}=S^b_{t}\vee\essinf[\cF_{t}]\tilde S^b_{t+1},
\end{align}
for $t=T-1,\dots,1$. Since $0<\tilde S^b_t\le \tilde S^a_t$ a.s. for all
$t$, \NRA\ always holds. By Definition \ref{def:nra-multiperiod} and
Lemma~\ref{lemma:NRAT}, we easily deduce the following result.
 
\begin{theorem}{\rm \quad}
  \label{NAS-2d}
  \label{NA2equiv}
  \begin{enumerate}[(i)]
  \item \SNR\ holds if and only if $S_t^b\le \esssup[\cF_t]\tilde
  S_{t+1}^a$  and $S_t^a\ge \essinf[\cF_t]\tilde S_{t+1}^b$ a.s. with strict inequalities when  $ S_{t}^b<S_{t}^a$,  for all $t\le T-1$.
  \item \NAT\ holds if and only if $\esssup[\cF_t]S_{t+1}^a\ge S^a_t$ and
    $S_t^b\ge \essinf[\cF_t] S_{t+1}^b$ a.s.  for all $t\le T-1$.
  \end{enumerate}
\end{theorem}

\begin{remark} 
  \NAS\ is equivalent to \SNR\ in the proper case but also to the
  existence of a strictly consistent price system, see
  \cite[Th.~3.2.2]{kab:saf09}. In the two asset case, the Grigoriev
  theorem, see \cite[Th.~3.2.15]{kab:saf09} and \cite{grig05}, asserts
  that \NAS\ is equivalent to the existence of a (possibly non-strict)
  consistent price system, i.e. the existence of a martingale $Z_t$
  with respect to a probability measure $\Q$ equivalent to $\P$ such
  that $S^b_t\leq Z_t\leq S^a_t$ for all $t$. By Theorem~\ref{NAS-2d},
  the existence of a consistent price system, equivalently, \NAS, implies
  \SNR. Indeed, $\esssup[\cF_t]\tilde S_{t+1}^a\ge
  \E_\Q(Z_{t+1}|\cF_t)\geq S^b_t$ and similarly $\essinf[\cF_t]\tilde
  S_{t+1}^b\le \E_\Q(Z_{t+1}|\cF_t)\leq S^a_t$, the inequalities being
  strict when $S^b_t<S^b_a$. 
\end{remark}

\begin{corollary}
  If there exist probability measures $\Q^a,\Q^b$ which are equivalent
  to $\P$, such that $S^a$ is a $\Q^a$-submartingale and $S^b$ is a
  $\Q^b$-supermartingale, then \NAT\ holds.
\end{corollary}

The condition in the following corollary means that $\delta_t^b\eqdef
S_t^b/S_{t-1}^b$ and $\delta_t^a\eqdef S_t^a/S_{t-1}^a$ admit
conditional full supports on $\R_+$ for all $t=1,\dots,T$.

\begin{corollary} 
  If $\P(\delta_t^b\le c|\cF_{t-1}) \P(\delta_t^a\ge c|\cF_{t-1})>0$
  a.s.  for all $t=1,\dots,T$ and all $c>0$, then \NAT\ holds.
\end{corollary}
\begin{proof}
  Let $\gamma=\esssup[\cF_{t-1}] S_t^a$. Then
  $\gamma\one_{\delta_t^a \geq c}\ge S_t^a\one_{\delta_t^a\geq
    c}\geq cS_{t-1}^a\one_{\delta_t^a\ge c}$. Taking the conditional
  expectation yields that
  \begin{displaymath}
    \gamma\P(\delta_t^a\ge c|\cF_{t-1})
    \ge cS_{t-1}^a\P(\delta_t^a\ge c|\cF_{t-1}).
  \end{displaymath}
  Then $\gamma\ge cS_{t-1}^a$, and letting $c\to\infty$ yields that
  $\gamma=+\infty$ a.s. Similarly, $\essinf[\cF_{t-1}] S_t^b=0$
  a.s., and Theorem~\ref{NA2equiv}(ii) applies.
\end{proof}

Assume now that the acceptability criterion is based on the
generalised conditional expectation as described in
Section~\ref{sec:sandwich-theorems}. By Proposition~\ref{prop:ce},
\NARA\ holds if and only if there is deterministic $z$ that belongs to
all $\bY_t$, $t=0,\dots,T$, and \SNR\ additionally requires that this
point belongs to the interiors of $\bY_t$.

\begin{example}[Limit order book]
  Consider the two asset setting, where it is allowed to perform
  transactions up to one cash unit amount. This is a simple limit
  order book setting with only one break point. Then $\bK_t$ is the
  sum $[0,\alpha_t]+[0,\beta_t]+\R_+^2$, where $\alpha_t=(1,-S^a_t)$
  and $\beta_t=(-1,S^b_t)$. By Proposition~\ref{SNR-CondExp}, \NARA\
  (with acceptability based on conditional expectation) 
  holds if and only if
  \begin{displaymath}
    \E\Big[\sum_{s=t+1}^T ([0,\alpha_s]+[0,\beta_s])\;\Big|\cF_t\Big]\cap \R_-^d=\{0\}
  \end{displaymath}
  for all $t=0,\dots,T-1$. The sum of segments $[0,\alpha_s]$ and
  $[0,\beta_s]$ is a random convex compact set called a zonotope. The
  setting can be easily extended to the case of limit order books with
  several break points.
\end{example}

\section{Appendix: Random sets and their selections}
\label{sec:meas-select-rand}

Let $\R^d$ be the Euclidean space with norm $\|\cdot\|$ and the Borel
$\sigma$-algebra $\cB(\XX)$. The closure of a set $A\subset\XX$ is
denoted by $\cl A$. 
A set-valued function $\omega\mapsto \bX(\omega)$ from a complete
probability space $(\Omega,\cF,\P)$ to the family of all subsets of
$\XX$ is called \emph{$\cF$-measurable} (or graph-measurable) if its
graph
\begin{displaymath}
  \Gr \bX\eqdef\big\{(\omega,x)\in \Omega\times \XX:
  x\in \bX(\omega)\big\}\subset \Omega\times\XX
\end{displaymath}
belongs to the product $\sigma$-algebra $\cF\otimes \cB(\XX)$.  In
this case, $\bX$ is said to be a \emph{random set}. In the same way
the $\cH$-measurability of $\bX$ with respect to a
sub-$\sigma$-algebra $\cH$ of $\cF$ is defined. 
The random set $\bX$ is said to be \emph{closed} (convex, open) if
$\bX(\omega)$ is a closed (convex, open) set for almost all $\omega$.

\begin{definition}
  An $\cF$-measurable random element $\xi$ in $\XX$ such that
  $\xi(\omega)\in\bX(\omega)$ for almost all $\omega\in\Omega$ is
  said to be an $\cF$-measurable \emph{selection} (selection in short)
  of $\bX$, $\Lnot(\bX,\cF)$ denotes the family of all
  $\cF$-measurable selections of $\bX$, and $\Lnot[p](\bX,\cF)$
  is the family of $p$-integrable ones.
\end{definition}

It is known that an a.s. non-empty random set has at least one
selection, see \cite[Th.~4.4]{hes02}. Let $\cH$ be a
sub-$\sigma$-algebra of $\cF$. 

\begin{definition}
  \label{def:decomp}
  A family $\Xi\subset \Lnot(\XX,\cF)$ is said to be
  \emph{infinitely $\cH$-decomposable} if
  \begin{math}
    \sum_n \xi_n\one_{A_n}\in \Xi
  \end{math}
  for all sequences $(\xi_n)_{n\geq1}$ from $\Xi$ and all
  $\cH$-measurable partitions $(A_n)_{n\geq1}$ of $\Omega$; $\Xi$ is
  \emph{$\cH$-decomposable} if this holds for finite partitions. 
\end{definition}

The decomposable subsets of $\Lnot(\R^d,\cF)$ are called stable and
infinitely decomposable ones are called $\sigma$-stable in
\cite{cher:kup:vog14}. The following result for $\cH=\cF$ is well
known in case $p=1$ \cite{hia:ume77}, where the decomposability
concept was first introduced; see also \cite[Th.~2.1.10]{mo1} for
$\cH=\cF$, and \cite[Prop.~5.4.3]{kab:saf09} 
for $p=0$.

\begin{theorem}[see \protect{\cite[Th.~2.4]{lep:mol19}} 
  and \protect{\cite[Th.~2.1.10]{mo1}}]
  \label{ReprDecompL0} 
  Let $\Xi$ be a non-empty subset of $\Lnot[p](\XX,\cF)$ for $p=0$ or
  $p\in[1,\infty]$. Then
  \begin{displaymath}
    \Xi\cap\Lnot[p](\XX,\cH)=\Lnot[p](\bX,\cH).
  \end{displaymath}
  for an $\cH$-measurable random closed set $\bX$ if and only if $\Xi$
  is $\cH$-decomposable and closed.
\end{theorem}

For $A_1,A_2\subset\XX$, define their elementwise (Minkowski) sum as
\begin{displaymath}
  A_1+A_2\eqdef\{x_1+x_2:\; x_1\in A_1,x_2\in A_2\}\,.
\end{displaymath}
The same definition applies to the sum of subsets of $\Lnot(\XX,\cF)$.
The set of pairwise differences of points from $A_1$ and $A_2$ is
obtained as $A_1+(-A_2)$, or shortly $A_1-A_2$, where
$-A_2\eqdef\{-x:\; x\in A_2\}$ is the centrally symmetric variant of
$A_2$.  For the sum $A+\{x\}$ of a set and a singleton we write
shortly $A+x$. Note that the sum of two closed sets is not necessarily
closed, unless at least one of the closed summands
is compact. The following result differs from \cite[Th.~1.3.25]{mo1}
in considering the possibly non-closed sum of two random closed sets.

\begin{lemma}
  \label{A+B-measurable} 
  Let $\bX$ and $\bY$ be two random sets. Then
  $\Lnot(\bX,\cF)+\Lnot(\bY,\cF)=\Lnot(\bX+\bY,\cF)$.
  If both $\bX$ and $\bY$ are random closed sets, then
  $\bX+\bY$ is measurable.
\end{lemma}
\begin{proof}
  It is trivial that $\Lnot(\bX,\cF)+\Lnot(\bY,\cF)\subseteq
  \Lnot(\bX+\bY,\cF)$. To prove the reverse inclusion,
  consider $\xi\in \Lnot(\bX+\bY,\cF)$. Since $\bX$
  and $\bY$ are $\cF$-measurable, the measurable selection
  theorem \cite[Th.~5.4.1]{kab:saf09} yields that there exist
  $\cF$-measurable selections $\xi'\in \Lnot(\bX,\cF)$ and
  $\xi''\in \Lnot(\bY,\cF)$ such that $\xi=\xi'+\xi''$. 

  Now assume that $\bX$ and $\bY$ are closed and consider
  their Castaing
  representations (see \cite[Def.~1.3.6]{mo1}) 
  $\bX(\omega)=\cl\{\xi'_i(\omega),i\geq1\}$ and
  $\bY(\omega)=\cl\{\xi''_i(\omega),i\geq1\}$. 
  The measurability of $\bX+\bY$ follows from the
  representation
  \begin{multline}
    \label{gr(A+B)}
    \Gr(\bX+\bY)\\
    =\bigcup_{k\geq1}\bigcap_{m\geq1}
    \bigcup_{i,j\geq1}\Big\{(\omega,x):\;\|x-\xi'_i(\omega)-\xi''_j(\omega)\|\le
    \frac{1}{m},\; \|\xi'_i(\omega)\|\le k\Big\}.
  \end{multline}
  Indeed, if $(\omega,x)\in \Gr(\Gamma_1+\Gamma_2)$, then $x=a+b$ for
  $a\in\Gamma_1(\omega)$ and $b\in\Gamma_2(\omega)$. Let $k\geq1$
  such that $\|a\|+1\le k$. Since $a\in \Gamma_1$, there exists a
  subsequence $(\xi'_{n_l})_{l\geq1}$ such that
  $\xi'_{n_l}(\omega)\to a$. We may assume without loss of generality
  that $\|\xi'_{n_l}(\omega)\|\le k$. Similarly,
  $\xi''_{n_l}(\omega)\to b$. Therefore, if $m>0$, then 
  $\|x-\xi'_i(\omega)-\xi''_j(\omega)\|\le \frac{1}{m}$ and
  $\|\xi'_i(\omega)\|\le k$ for some $i,j$.
\end{proof}

\begin{acknowledgements}

  IM was supported by the Swiss National Science Foundation Grant
  200021-153597.
 \\
 EM thanks the program {\sl Investissements d'Avenir} from the French foundation ANR which supports the  Bachelier colloquium, Metabief, France. 
\end{acknowledgements}


\end{document}